%% file: main.tex
\documentclass[12pt]{amsart}

\usepackage[top=1in,bottom=1in,left=1in,right=1in]{geometry}

\input{ams_header}

\title{A commuting operator self-test for exact entanglement embezzlement}

\author{Connor Paddock}
\author{Simon Schmidt}

\address{Connor Paddock, Department of Computer Science and the Institute for Quantum Science and Technology, University of Calgary, Canada}
\email{connor.paddock@ucalgary.ca}

\address{Simon Schmidt, Faculty of Computer Science, Ruhr University Bochum, Germany}
\email{s.schmidt@rub.de}

\begin{document}

\begin{abstract}
We show that the optimal commuting operator correlation for a variant of Coladangelo's two-player generalized nonlocal game, based on the embezzlement of entanglement, is a commuting operator self-test. Our result gives the first correlation self-test for an exact entanglement embezzlement protocol in the commuting operator framework. Because the unique optimal correlation necessitates the exact embezzlement of entanglement, the correlation is unrealizable using quantum tensor-product models, even if the local spaces are allowed to be infinite dimensional. Our proof relies on Harris' unique abstract state characterization of exact embezzlement of entanglement involving Cuntz algebras.
\end{abstract}

\maketitle

\section{Introduction}

In a generalized two-player nonlocal game, a verifier sends questions $x\in X$ and $y\in Y$ to two noncommunicating players, Alice and Bob, according to a known distribution. They both return answers $a\in A$ and $b\in B$, and receive a score based on their answers. Their behaviour is described by a correlation $p:=\{p(a,b|x,y)\}_{a,b,x,y}$, and their expected score is a linear function of $p$. Quantum strategies allow the players to share an entangled state and perform local measurements.

More specifically, if the shared state is $|\psi\rangle\in H_A\otimes H_B$ and the players employ POVMs $\{\{M_a^x\}_{a\in A}:x\in X\}$ on $H_A$ and $\{\{N_b^y\}_{b\in B}:y\in Y\}$ on $H_B$, then we have 
\begin{equation*}
    p(a,b|x,y)=\langle \psi|M_a^x\otimes N_b^y|\psi\rangle,\quad \text{for all $(x,y,a,b)\in X\times Y \times A \times B$}.
\end{equation*} Each quantum correlation $p$ is determined by a quantum model (or strategy) $S$, consisting of a finite-dimensional state and measurement operators. In fact, a correlation is a \emph{quantum correlation} precisely if such a model $S$ exists. For fixed sets $X,Y,A,$ and $B$, we define the set of classical (resp. quantum) correlations by $C_c(X,Y,A,B)$ (resp. $C_q(X,Y,A,B)$). Notably, $C_q(X,Y,A,B)$ contains $C_c(X,Y,A,B)$, and is a convex but generally not closed subset of $\R^{X\times Y\times A \times B}$ \cite{slofstra2019set}.

Most quantum correlations admit multiple quantum models. However, there are special correlations which have a unique quantum model. These correlations are called \emph{self-tests}, and they allow an observer to infer the players' vector state and measurements. This makes self-testing an important tool in the device independent certification of quantum devices \cite{mayers1998quantum, mayers2003self, reichardt2013classical}.

In self-testing, the uniqueness of the quantum model $S$ for $p$ is conventionally defined in terms of the existence of $(i)$ an \emph{ideal model} $\wtd{S}$ and $(ii)$ local isometries taking any other model $S$ to the ideal model $\wtd{S}$ with an auxiliary state. More recently, a succinct operator algebraic definition of self-testing was provided in terms of abstract states \cite{PSZZ}. The abstract state definition of self-testing is motivated by the fact that if $p\in C_q$ is an extreme point, then $p$ is a self-test if and only if all mixed moments of the form
\[
\langle \psi|M_{a_1}^{x_1}\cdots M_{a_k}^{x_k}\otimes N_{b_1}^{y_1}\cdots N_{b_\ell}^{y_\ell}|\psi\rangle,
\]
are equal for every quantum model realizing $p$ \cite[Theorem 3.5]{PSZZ}. We call the latter condition an \emph{abstract state self-test} and it is equivalent to the property that there is a unique abstract finite-dimensional state $f:\mcA_{\mathrm{POVM}}^{(X,A)}\otimes_{\min} \mcA_{\mathrm{POVM}}^{(Y,B)}\to \C$ such that $p(a,b|x,y)=f(m_a^x\otimes n_b^y)$. Here, $\mcA_{\mathrm{POVM}}^{(X,A)}$ is the universal $C^*$-algebra of POVMs with generators $\{m_a^x\}$ satisfying the relations $m_a^x\geq 0$ and $\sum_a m_a^x=1$. Abstract states $f$ are unital, positive linear functionals on $C^*$-algebras. The concrete connection between abstract state self-testing and the conventional definition can be seen through the Gelfand-Naimark-Segal (GNS) construction. If $f$ is the  unique abstract state realizing $p$, and the Hilbert space in the GNS construction is finite-dimensional, then the GNS representation and cyclic vector provide the ideal quantum model $\wtd{S}$.

The abstract state definition naturally generalizes to other families of correlations. For instance, if we drop the finite-dimensional requirement on the state, then a correlation $p$ belongs to the closure of $C_q$ (denoted by $C_{qa})$ if and only if there is a state $f$ on $\mcA_{\mathrm{POVM}}^{(X,A)}\otimes_{\min} \mcA_{\mathrm{POVM}}^{(Y,B)}$ realizing $p$ \cite{fritz2012tsirelson}. This is not the same as dropping the finite-dimensional requirement on the local Hilbert spaces. Dropping this assumption gives another class $C_{qs}$, which denotes the correlations admitting a tensor-product model with possibly infinite-dimensional local spaces. It has been shown that $C_{qs}$ is not closed, yet $\overline{C_{qs}}=C_{qa}$, and therefore correlations in $C_{qa}$ need not admit such models.

This lack of tensor-product realizations informs the more general notion of commuting operator correlations. A correlation $p$ is a commuting operator correlation if and only if there is a state $f$ on $\mcA_{\mathrm{POVM}}^{(X,A)}\otimes_{\max} \mcA_{\mathrm{POVM}}^{(Y,B)}$ realizing $p$. We let $C_{qc}$ denote the set of commuting operator correlations. States on the $\max$ tensor-product have GNS representations with a state on a single Hilbert space $H$ and commuting POVM operators rather than on separate factors of a tensor-product. Every state on $\mcA_{\mathrm{POVM}}^{(X,A)}\otimes_{\min} \mcA_{\mathrm{POVM}}^{(Y,B)}$ pulls back to a state on $\mcA_{\mathrm{POVM}}^{(X,A)}\otimes_{\max} \mcA_{\mathrm{POVM}}^{(Y,B)}$, matching the fact that $C_{qa}\subset C_{qc}$\footnote{The landmark MIP$^*=$RE work of Ji, Natarajan, Vidick, Wright, and Yuen actually shows that $C_{qa}\subsetneq C_{qc}$ \cite{ji2021mip}.}.

The abstract state self-testing framework provides a natural definition for commuting operator correlations. That is, we say $p$ is a \emph{commuting operator self-test} if there is a unique abstract state $f$ on $\mcA_{\mathrm{POVM}}^{(X,A)}\otimes_{\max} \mcA_{\mathrm{POVM}}^{(Y,B)}$ realizing $p$. In \cite{PSZZ} several correlations were shown to be commuting operator self-tests, including the optimal correlation for the CHSH game. However, in these commuting operator self-tests, the GNS Hilbert space of the unique state was finite-dimensional (and therefore the correlations belonged ultimately to $C_q$). This raised the question of whether one could find a commuting operator self-test that did not belong to $C_q$ or even $C_{qs}$.

Coladangelo’s generalized two-player nonlocal game \cite{coladangelo2020two} (and the follow-up \cite{beigi2021separation}) provides a natural candidate. The game \(G_{\mathrm{emb}}\) combines two commuting-operator self-tests with a consistency test. Individually, the subgames each self-test the states
\[
|\Phi_3\rangle=\frac{|00\rangle+|11\rangle+|22\rangle}{\sqrt3},
\quad\text{ and}\quad
|\Psi_2\rangle=\frac{|00\rangle+\sqrt2|11\rangle}{\sqrt3},
\]
respectively\footnote{We replaced one of the subgames from the original description in \cite{coladangelo2020two}, using the commuting operator self-test for $|\Phi_3\rangle$ from \cite{schmidt2025guess} instead.}. By forcing alignment of the \(|00\rangle\) components, the consistency test prevents the strategy of playing each subgame independently with the state $|\Phi_3\rangle \otimes |\Psi_2\rangle$. Comparing the remaining components, one can show that a coherent transformation from \(|11\rangle\) to \(\tfrac{1}{\sqrt{2}}(|11\rangle+|22\rangle)\) combined with the optimal subgame strategies would achieve the optimal quantum score for $G_{emb}$. Coladangelo shows that there is a sequence of finite-dimensional strategies which approach the quantum value of the game, despite no tensor-product strategy attaining it exactly. Hence, the resulting optimal correlation is in $C_{qa}\setminus C_{qs}$, and the same is true for our game $G_{emb}$. We will show that this optimal commuting operator correlation is unique and a commuting operator self-test.

The starting point, which is already used extensively in \cite{coladangelo2020two}, is that the coherent transformation above resembles the embezzlement of a maximally entangled qubit state. Indeed, the approximate embezzlement of entanglement \cite{van2003universal} informs the sequence of finite-dimensional strategies mentioned above. However, unlike in \cite{coladangelo2020two} we focus here on the case of exact embezzlement. In \cite{cleve2017perfect}, Cleve, Liu, and Paulsen showed that the exact embezzlement of entanglement is possible in the commuting operator framework, while simultaneously establishing impossibility in the tensor-product framework, even with infinite-dimensional local spaces.

In this paradigm, spatially separated parties begin with a product state between local registers. In addition, they share a \emph{catalyst} state $|\Psi\rangle$ on a ``global'' Hilbert space, which acts as a resource. Their ``local'' operations must commute on the ``global'' space, and their goal is to create an entangled state on the local registers while preserving the catalyst state. The exact embezzlement of entanglement for the maximally entangled two-qubit state is the transformation $|00\rangle \otimes |\Psi\rangle \mapsto \tfrac{1}{\sqrt{2}}(|00\rangle+|11\rangle) \otimes |\Psi\rangle$. Letting $K$ be a potentially infinite dimensional Hilbert space, the transformation is described by a pair of unitaries $U_A\in \mcB(\C^2\otimes K)$ and $U_B\in \mcB(K\otimes \C^2)$ such that $[U_A\otimes \Id_{2},\Id_{2}\otimes U_B]=0$, along with a catalyst state $|\Psi\rangle \in K$ for which
\[
    (U_A\otimes \Id_2)(\Id_2\otimes U_B)|0\rangle |\Psi\rangle |0\rangle=\tfrac{1}{\sqrt{2}}|0\rangle |\Psi\rangle |0\rangle+\tfrac{1}{\sqrt{2}} |1\rangle |\Psi\rangle |1\rangle.
\]

The existence of an exact embezzlement protocol in the commuting operator framework uses similar concepts from $C^*$-algebras to those used to characterize commuting operator correlations. Like the POVM algebra, the \emph{universal algebra of commuting $2\times 2$ unitaries} $\mcU_2\otimes_{\max}\mcU_2$ encodes the abstract ``embezzling unitaries''. The existence of an exact embezzlement protocol comes down to whether there is a vector state and a suitable representation satisfying the equation above. Notably, if there is a state on $\mcU_2\otimes_{\max}\mcU_2$ such that $f(u_{i0}\otimes u_{j0})=\delta_{ij}/\sqrt{2}$, then by the GNS construction, there is a Hilbert space $H_f$, a representation $\pi_f$, which provides unitaries
\begin{equation*}
    U_A=\sum_{ij}\ket{i}\bra{j}\otimes \pi_f(u_{ij}\otimes 1),\;\text{ and }\,U_B=\sum_{ij} \pi_f(1\otimes u_{ij})\otimes \ket{i}\bra{j},
\end{equation*}
and a catalyst $|\Omega_f\rangle\in H_f$. Together, they provide a ``concrete'' exact embezzlement protocol from the abstract state $f$. Conversely, it was shown in \cite{cleve2017perfect} that the limit point of the approximate embezzlement (finite-dimensional) states provides such an abstract state. If the embezzling state were unique, it would suggest that the optimal correlation for $G_{emb}$ is a commuting operator self-test. However, it is not hard to see that the embezzlement relations do not determine a unique state on $\mathcal U_2\otimes_{\max}\mathcal U_2$. Non-uniqueness of the state stems from the fact that the output register is initialized to $|00\rangle$, so the relations involve only the first columns of the implementing unitaries, leaving parts of the protocol unconstrained. Nonetheless, it remains plausible that the state is unique after accounting for this degree of freedom.

In \cite{harris2026self}, Harris identified this unique component of the protocol using Cuntz algebras. Briefly, the Cuntz algebra $\mcO_2$ is the universal $C^*$-algebra of isometries $v_0$ and $v_1$ satisfying $v_0v_0^*+v_1v_1^*=\Id$. In particular, Harris showed that after suitable compressions, the adjoints of the first-column entries of the embezzling unitaries generate a Cuntz algebra, and moreover, the embezzlement relations determine a unique state on the tensor-product $\mcO_2\otimes_{\max} \mcO_2$.

To apply this uniqueness result to our nonlocal game, we first establish that the Cuntz relations manifest within any optimal strategy of $G_{emb}$. We begin with the two self-testing subgames. In each, the local measurement algebras of the players are copies of $M_3(\mathbb C)$ and $M_2(\mathbb C)$ respectively in any optimal strategy. The consistency game identifies a common projection within each player's local algebra, corresponding to outcome $|0\rangle$. Compressing the generated algebra to the image of the common projection we find a copy of the Cuntz algebra $\mathcal O_2$. Furthermore, the joint state on $\mathcal O_2\otimes_{\max}\mathcal O_2$ induced by an optimal strategy satisfies the embezzlement relations, and hence is uniquely determined by the result of Harris. To obtain a commuting operator self-test it remains to recover the full state on the POVM algebras from this compressed state. However, the rigid relations supplied by the subgames allow us to extend the compressed state uniquely. Hence, the uniqueness on the tensor-product of Cuntz algebras determines the entire abstract state, establishing both uniqueness of the optimal correlation, and the commuting operator self-testing property.

    \begin{theoremI}[\cref{thm:comopselftest}]\label{thm:main}
        The generalized nonlocal game $G_{emb}$ has a unique optimal commuting operator correlation $\hat{p}$. This correlation is a commuting operator self-test and belongs to $C_{qa}(4,4,3,3)\setminus C_{qs}(4,4,3,3)$.
    \end{theoremI}

    
The result provides an example of an abstract state self-test that admits neither a finite-dimensional nor a tensor-product realization, resolving the open question posed in \cite{PSZZ}. As a corollary, we observe that every optimal strategy achieving the optimal score in $G_{emb}$ gives rise to an exact embezzlement protocol for the two-dimensional maximally entangled state.

\begin{corI}[\cref{cor:embez_protocol}]
    Every optimal cyclic model for $G_{\mathrm{emb}}$ determines an exact embezzlement protocol for
    \[|\Phi_2\rangle=\frac{|00\rangle+|11\rangle}{\sqrt2}.\]
    The extracted catalyst and local operators are unique up to unitary equivalence.
\end{corI}

By observations of \cite{van2025multipartite} (and \cite[Example 4.13]{harris2026self}), it follows that in the unique model, each player’s measurements generate a type $\mathrm{III}_{1/2}$ von Neumann factor. This raises a natural question: can nonlocal games self-test exact embezzlement of other states and thereby certify further types of von Neumann factors? In particular, \cite[Theorem 3.4]{harris2026self} establishes uniqueness for exact embezzlement of arbitrary bipartite states. Each abstract state $f:\mcO_d\otimes_{\max} \mcO_d\to \C$ satisfying $f(v_i\otimes v_j)=\delta_{ij}\alpha_i$ for $i,j=0,\ldots,d-1$ uniquely determines an exact embezzlement protocol for a quantum state with Schmidt coefficients $(\alpha_0,\dots, \alpha_{d-1})$.

 \section*{Acknowledgements}

 Connor Paddock acknowledges early discussions with Yuming Zhao and William Slofstra pertaining to the correlation in \cite{coladangelo2020two} during the completion of \cite{PSZZ}. He would also like to thank Lauritz van Luijk for answering some questions regarding embezzlement and von Neumann algebras, and pointing him to the reference \cite{van2025multipartite}. Both authors thank Samuel Harris for helpful conversations regarding the results in \cite{harris2026self}.
 Simon Schmidt acknowledges support from the Deutsche Forschungsgemeinschaft (DFG, German Research Foundation) under Germany’s Excellence Strategy – EXC 2092
CASA – 390781972. Connor Paddock acknowledges support from the University of Calgary's Faculty of Science and Quantum City startup awards.

 \section*{Disclosure of AI usage}

 AI tools (ChatGPT 5.6/6.0) were used to aid the authors in parts of the work. However, the authors came up with the main ideas for the project, and developed the approach of combining the results of \cite{harris2026self} and \cite{coladangelo2020two} to obtain the main result (\cref{thm:comopselftest}). AI tools were primarily employed in the development of \cref{lem:isoCuntz}. They were also used in the proofs of \cref{lem:projective_to_POVM}, \cref{thm:tCHSH}, \cref{lem:uniquestate}, and \cref{prop:optimal_state}, whereas the authors were responsible for the claims themselves. Otherwise, AI tools were limited to minor aspects like finding typos and grammatical errors. The authors take full responsibility for the correctness of the work.

\section*{Concurrent work}

During the completion of this work the authors became aware of \cite{KSZ26} in which the authors claim to construct a synchronous correlation $p$ that is both a commuting operator self-test and $p\in C^{sync}_{qa}\setminus C^{sync}_{qs}$.

\section{Preliminaries}\label{sec:prelims}

\subsection{Mathematical background and notation}

Let $H$ be a Hilbert space. In this work, all Hilbert spaces will be over $\C$. We denote by $\mbB(H)$ the space of bounded linear operators on $H$, with unit $\Id$. A unital $C^*$-algebra $\mcA$ is a complex Banach $*$-algebra with a unit $1_\mcA$, where the norm satisfies the identity $\|aa^*\|=\|a\|^2$, for all $a\in \mcA$. All of the $C^*$-algebras will be unital in this work. The obvious example of a $C^*$-algebra is $\mbB(H)$. In fact, by the Gelfand-Naimark theorem, every $C^*$-algebra is isometrically $*$-isomorphic to a norm-closed self-adjoint subalgebra of some $\mbB(H)$. If $\mcF\subseteq \mbB(H)$, we write $C^*(\mcF)$ to denote the smallest norm-closed $*$-subalgebra of $\mbB(H)$ containing $\mcF$. Equivalently, $C^*(\mcF)=\overline{\mathrm{span}}^{\|\cdot\|}\{F_1\cdots F_n:n\geq 0, F_i\in \mcF\cup \mcF^*\}$, where $\Id$ denotes the empty product, and $\mcF^*=\{F^*:F\in \mcF\}$.

On the other hand, given just a set of generators $X$ and admissible relations $R$, we let $C^*(X,R)$ be the \emph{universal $C^*$-algebra} of $(X,R)$. In particular, $C^*(X,R)$ is the completion of the quotient of the free $*$-algebra $\C^*\ang{X}$ with respect to the ideal $\{a: \|a\|_U=0\}$, where $\|a\|_U=\sup_\rho\|\rho(a)\|$ is a seminorm induced by the $*$-representations $\rho:\C^*\ang{X}\to \mbB(H)$ satisfying the relations $R$. This is particularly useful for abstractly defining the notion of quantum measurements.

Measurements are described by \emph{positive operator valued measures} (POVMs), which consist of a family of positive operators $\{M_i \in \mbB(H)\,|\, i \in [m]\}$ such that $\sum_{i=1}^m M_i= \Id$. If all positive operators are projections ($M_i=M_i^*=M_i^2$), then we call $\{M_i \in \mbB(H)\,|\, i \in [m]\}$ with $\sum_{i=1}^m M_i= \Id$ a \emph{projective measurement} (PVM).

Using the earlier idea, we define $\mcA_{\mathrm{POVM}}^{(X,A)}$ to be the universal POVM ($C^*$)-algebra with generators $\{\{m_a^x\}_{a\in A}:x\in X\}$ and relations $m_a^x\geq 0$ for all $a\in A, x\in X$ and $\sum_am_a^x=1$ for all $x\in X$. Similarly, we let $\mcA_{\mathrm{PVM}}^{(X,A)}$ be the universal PVM (C$^*$)-algebra with generators $\{\{p_a^x\}_{a\in A}:x\in X\}$ and relations $p_a^x=(p_a^x)^*=(p_a^x)^2$ for all $a\in A, x\in X$ and $\sum_a p_a^x=1$ for all $x\in X$. The universal property ensures that every representation of $\mcA_{\mathrm{POVM}}^{(X,A)}$ (resp.~$\mcA_{\mathrm{PVM}}^{(X,A)}$) corresponds to a collection of quantum measurements (resp.~projective measurements)  and vice versa.

There is another universal $C^*$-algebra that we will encounter in this work. We define the \emph{Cuntz algebra} $\mcO_d$ to be the universal $C^*$-algebra generated by isometries $V_0,\ldots, V_{d-1}$ satisfying $\sum_{i=0}^{d-1}V_iV_i^*=1$. For $d\geq 2$ the algebra $\mcO_d$ is simple, and therefore any isometries $\wtd{V}_0,\ldots, \wtd{V}_{d-1} \in \mbB(H)$ satisfying $\sum_{i=0}^{d-1}\wtd{V_i}\wtd{V_i}^*=\Id$ generate a $C^*$-algebra isomorphic to $\mcO_d$.

A \emph{state} $f:\mathcal A\to \C$ on a $C^*$-algebra $\mathcal A$ is a positive linear functional with $f(1_{\mathcal A})=1$. The state space $\mcS(\mcA)$ of a $C^*$-algebra is both convex and weak-$*$ compact. A (quantum) vector state $\ket{\psi}$ is a unit vector in a Hilbert space $\mathcal{H}$. Note that any vector state $\ket{\psi}$ gives rise to a state $\Psi$ on $\mbB(H)$ via $\Psi(\cdot)=\tr(\cdot\ket{\psi}\bra{\psi})$. To avoid confusing vector states in $H$ with states on $C^*$-algebras, we use the term \emph{abstract state} for the latter. Abstract states also give rise to vector states. This is the famous GNS construction of Gelfand, Naimark, and Segal. Indeed, the \emph{GNS construction} says that every (abstract) state $f$ on a $C^*$-algebra gives rise to a vector state on a Hilbert space, such that $\langle \Omega|\pi(a)|\Omega\rangle=f(a)$ for all $a \in \mcA$. Notably, the Hilbert space, (unit) vector, and representation $(H,\pi, |\Omega\rangle)$ are the GNS triple associated with each $f$. The GNS vector $|\Omega\rangle$ is \emph{cyclic}, meaning that $\pi(\mcA)|\Omega\rangle$ is dense in $H$. 

Given two $C^*$-algebras $\mcA$ and $\mcB$ there are two different ways to complete the tensor-product $\otimes_{alg}$ to obtain a new $C^*$-algebra. The first is the $C^*$-algebra $\mcA \otimes_{\min}\mcB$,  known as the $\min$ construction. In the min construction, the norm is completed with respect to the seminorm induced by (spatial) representations $\pi_A\otimes \pi_B$, where $\pi_A:\mcA\to \mbB(H_A)$, and $\pi_B:\mcB\to \mbB(H_B)$. The other is the $C^*$-algebra $\mcA \otimes_{\max}\mcB$, known as the $\max$ construction. Here the completion is taken with respect to the seminorm induced by commuting representations $\pi:\mcA\otimes_{alg} \mcB \to\mbB(H)$, where $[\pi(a\otimes 1),\pi(1\otimes b)]=0$ for all $a\in \mcA,b\in \mcB$.

\subsection{Commuting operator self-tests}

Let $A,B,X$, and $Y$ be finite sets, $\mu\colon X\times Y\to \R_{\geq 0}$ be a probability distribution, and $V\colon A\times B \times X\times Y \to \R$ be a function. Note that nonlocal games are usually defined via predicates $V$, with $V(a,b|x,y) \in \{0,1\}$: we use the more general definition from \cite{coladangelo2020two} involving a score.

\begin{definition}\label{def:nonlocal_game}
A \emph{(generalized, two-player) nonlocal game} is a tuple~$\mcG=(A,B,X,Y,\mu,V)$ describing a scenario consisting of non-communicating players, Alice and Bob, interacting with a referee.
In the game, the referee samples a pair of questions $(x,y)\in X\times Y$ according to~$\mu$, and sends question~$x$ to Alice and~$y$ to Bob.
Then, Alice (resp.\ Bob) returns answer~$a$ (resp.\ $b$) to the referee. Then, the players obtain the score $V(a,b|x,y)$ (in ordinary nonlocal games, we say that they win if $V(a,b|x,y)=1$, otherwise they lose).
\end{definition}

The players are allowed to decide on a strategy before the game starts.
However, once the game begins they are not allowed to communicate.
To the referee, the behaviour of the players can be modelled by the probabilities~$p(a,b|x,y)$ of answers~$a,b$ given questions~$x,y$ as determined by the strategy.
The collection of numbers~$\{p(a,b|x,y)\}_{a\in A,b\in B,x\in X,y\in Y} \in \R^{A \times B \times X \times Y}$ is called a \emph{(bipartite) correlation}.
Thus, the expected score (probability of winning the game $\mcG$ for ordinary nonlocal games) under a strategy~$S$, with correlation~$p$, is given by
\begin{equation*}
    \omega(\mcG,S)
=   \sum_{x \in X, y \in Y}\sum_{a \in A, b \in B}\mu(x,y)V(a,b|x,y)p(a,b|x,y).
\end{equation*}

Note that in this formulation, (generalized) nonlocal games and Bell inequalities are equivalent concepts. Quantum strategies for such games are modelled by the players being able to use POVMs on a shared (entangled) state. The most general definition is the following.

\begin{definition}
    A \emph{commuting operator model} $S$ for a correlation $p$ consists of:
    \begin{itemize}
        \item[-] a Hilbert space $H$,
        \item[-] collections of measurement operators (a.k.a~POVMs) $\{M_a^x\}$ and $\{N_b^y\}$ such that $[M_a^x,N_b^y]=0$ for all $(x,y,a,b)\in X \times Y \times A\times B$, and
        \item[-] a vector state $|\psi\rangle \in H$,
    \end{itemize}
    such that $$p(a,b|x,y)=\langle \psi|M_a^x\cdot N_b^y|\psi\rangle,$$
    for all $(x,y,a,b)\in X \times Y \times A\times B$.
\end{definition}

In the context of a (generalized) nonlocal game a commuting operator model is the same thing as a commuting operator strategy. The \emph{commuting operator value} $\omega_{qc}$ of the game is then given by the supremum of the expected score of all such strategies.

\begin{definition}
    A commuting operator model $(H,\{M_a^x\},\{N_b^y\},|\psi\rangle)$ is \emph{cyclic} if $$H=\overline{C^*(\{M_a^x\},\{N_b^y\})|\psi\rangle}.$$
    Observe that every commuting operator model restricts to a cyclic model on the Hilbert space $\widehat{H}=\overline{C^*(\{M_a^x\},\{N_b^y\})|\psi\rangle}$.
    In particular, the GNS model of a state $f$, denoted by $S_f=(H_f,\{\pi_f(m_a^x\otimes 1)\},\{\pi_f(1\otimes n_b^y)\},|\Omega_f\rangle)$, is cyclic on $H_f$.
    Moreover, we say that $S$ is \emph{locally cyclic} if $$H=\overline{C^*(\{M_a^x\})|\psi\rangle}=\overline{C^*(\{N_b^y\})|\psi\rangle}.$$
\end{definition}

We remark that a commuting operator model is locally cyclic if its state vector is cyclic for each local measurement algebra separately. This condition is a commuting operator generalization of the full Schmidt rank condition for (quantum) tensor-product models whose local measurement algebras are full matrix algebras, and is close to the ``marginal cyclic'' condition introduced by Crann-Todorov-Turowska \cite{CTT}.

We now recall the following key definition.

\begin{definition}[\cite{PSZZ}, Definition 7.1]
    A correlation $p$ is a \emph{commuting operator self-test} if there exists a unique state $f$ on $\mathcal A_{\mathrm{POVM}}^{X,A}\otimes_{max}\mathcal A_{\mathrm{POVM}}^{Y,B}$ with $f(m_{a}^x\otimes n_{b}^y)=p(a,b|x,y)$ for  all $(x,y,a,b)\in X\times Y\times A \times B$.
\end{definition}

A common method for establishing a self-test is by showing that the game functional has a unique optimal state. This determines both the optimal correlation and the unique abstract state realizing it. The following result shows that proving uniqueness among pure optimal states suffices.

\begin{proposition}\label{prop:optimal_state}
    Let $\mcS$ be the state space of $\mcA_{\mathrm{POVM}}^{(X,A)}\otimes_{\max} \mcA_{\mathrm{POVM}}^{(Y,B)}$, and let $F=F^*\in\mcA_{\mathrm{POVM}}^{(X,A)}\otimes_{\max} \mcA_{\mathrm{POVM}}^{(Y,B)} $ be a self-adjoint element. Write
    \begin{equation*}
        \lambda=\max_{f\in \mcS} f(F),
    \end{equation*}
   and suppose there is a state $\widehat{f}$ such that every pure state $f$, satisfying $f(F)=\lambda$, equals $\hat{f}$. Then, $\hat{f}$ is the unique optimal state of $F$, and the correlation
   \begin{equation*}
       \hat{p}(a,b|x,y)=\hat{f}(m_a^x\otimes n_b^y),
   \end{equation*}
   is the unique optimal commuting operator correlation for $\lambda$. Moreover, $\hat{f}$ is the unique state in $\mcS$ realizing $\hat{p}$.
\end{proposition}

\begin{proof}
    The maximizing states of $F$ form a nonempty compact face $\mcF\subseteq \mcS$. Moreover, each extreme point of $\mcF$ is pure, and by our hypothesis, equals $\hat{f}$. By the Krein-Milman theorem, $\mcF$ is the closed convex hull of its extreme points, hence $\mcF=\{\hat{f}\}$. Since every optimal correlation is induced by a state in $\mcF$, every optimal correlation equals $\hat{p}$. Conversely, any state realizing $\hat{p}$ obtains $\lambda$, hence it belongs to $\mcF$, and therefore equals $\hat{f}$.
\end{proof}

\begin{remark}
    In \cref{prop:optimal_state}, observe that commuting operator correlations are particular subsets of the degree 2 moments of the elements of the states $f\in \mcS$. Otherwise, there is nothing special about the choice of $C^*$-algebra $\mcA$ in the claim. Indeed, the unique state on $\mcS(\mcA)$ conclusion holds if $F$ is replaced by any self-adjoint element of a $C^*$-algebra $\mcA$. In particular, if we used the PVM algebra instead, we can conclude that $\hat{f}$ is the unique projective state obtaining $\lambda$.
\end{remark}

It follows that establishing a commuting operator self-test from a nonlocal game or Bell scenario involves finding an \emph{ideal or optimal} model $\hat{f}$, and then proving that all states attaining the optimal value are equal to $\hat{f}$. A useful technique employed here is to find a sum-of-squares decomposition for quantity $\lambda-F$ inside of the POVM (or PVM) algebra to derive relations on the elements in the case where $f(F)=\lambda$ for all GNS representations.

The following lemma can be seen as a certain generalization of \cite[Lemma 5.1]{PSZZ}. It extends the result not only from finite to infinite dimensions but also from the synchronous case to the locally cyclic case. More generally, however, we do not know the full extent to which PVM self-testing suffices for POVM self-testing; see \cite{BCKLMNS} for a comparison of the definitions. 

\begin{lemma}\label{lem:projective_to_POVM}
    Let $p$ be a correlation in $C_{qc}(X,Y,A,B)$. If there is a unique state $f$ on $\mcA_{\mathrm{PVM}}^{(X,A)}\otimes_{\max} \mcA_{\mathrm{PVM}}^{(Y,B)}$ realizing $p$ with GNS representation $(\pi_f,H_f,|\Omega_f\rangle)$, and $f$ is locally cyclic
    \begin{equation*}
        \overline{\pi_f(\mcA_\mathrm{PVM}^{(X,A)}\otimes \mcI)|\Omega_f\rangle}=\overline{\pi_f(\mcI\otimes\mcA_\mathrm{PVM}^{(Y,B)})|\Omega_f\rangle}=H_f.
    \end{equation*}
    Then every cyclic commuting POVM model realizing $p$ is a PVM model. Consequently, the unique state realizing $p$ is $f\circ q$, where \begin{equation*}
        q:\mcA_{\mathrm{POVM}}^{(X,A)}\otimes_{\max} \mcA_{\mathrm{POVM}}^{(Y,B)}\onto \mcA_{\mathrm{PVM}}^{(X,A)}\otimes_{\max} \mcA_{\mathrm{PVM}}^{(Y,B)}
    \end{equation*}
    is the quotient map.
\end{lemma}

Before we give the proof we recall the following result.

\begin{lemma}[Naimark dilation for commuting operators]\label{lem:naimark}
Let $H$ be a Hilbert space and let $\{\{M_a^x:a\in A\},x\in X\}$, and $\{\{N_b^y:b\in B\},y\in Y\}$ be POVMs on $H$. Suppose
\begin{equation*}
    [M_a^x,N_b^y]=0\quad\text{for all $x,y,a,b\in X\times Y\times A \times B$}.
\end{equation*}
Write $\mcM_A=C^*(\{M_a^x\})$ and $\mcM_B=C^*(\{N_b^y\})$ for the local measurement algebras. There exist finite-dimensional Hilbert spaces $K_A$, and $K_B$, unit vectors $|\xi_A\rangle\in K_A$, $|\xi_B\rangle\in K_B$, and PVMs $\{\{P_a^x:a\in A\},x\in X\}$, and $\{\{Q_b^y:b\in B\},y\in Y\}$ on $H\otimes K_A\otimes K_B:=\wtd{H}$, such that
\begin{equation*}
    P_a^x\in \mcM_A\otimes \mbB(K_A)\otimes \Id_{K_B}\quad\text{and } Q_b^y\in \mcM_B\otimes \Id_{K_A}\otimes \mbB(K_B).
\end{equation*}
Moreover, for the isometry $W:H\to\wtd{H}$, with $W|v\rangle=|v\rangle\otimes |\xi_A\rangle\otimes |\xi_B\rangle$, we have that
\begin{equation*}
    W^*P_a^xW=M_a^x, W^*Q_b^yW=N_b^y,\text{ and } W^*P_a^xQ_b^yW=M_a^xN_b^y
\end{equation*}
$x,y,a,b\in X\times Y\times A \times B$. In particular, this implies $[P_a^x,Q_b^y]=0$, for all $x,y,a,b\in X\times Y\times A \times B$. Lastly, one may take $\dim(K_A)=|A|+1$, and $\dim(K_B)=|B|+1$.
\end{lemma}

The proof is standard; see, for instance \cite[Proof of Theorem 5.3]{paulsen2016estimating}

\begin{proof}[Proof of \cref{lem:projective_to_POVM}]
Let $S=(H,\{M_a^x\},\{N_b^y\},|\Omega\rangle)$ be a cyclic commuting POVM model realizing $p$. Define the local measurement algebras $\mcM_A=C^*(\{M_a^x\})$ and $\mcM_B=C^*(\{N_b^y\})$. Now, by \cref{lem:naimark}, the Naimark dilation $\wtd{S}$ of $S$ is a PVM model $(\tilde{H},\{P_a^x\},\{Q_b^y\},W|\Omega\rangle)$ realizing $p$.
Let \begin{equation*}
    K=\overline{C^*(\{P_a^x\},\{Q_b^y\})W|\Omega\rangle}.
\end{equation*}
Then $\wtd{S}|_K$ is a cyclic projective commuting model for $p$ with induced state $f$. By uniqueness of $f$ and the locally cyclic condition for $S_f$ we have
\begin{equation*}
    \overline{C^*(\{P_a^x\})W|\Omega\rangle}=
    \overline{C^*(\{Q_b^y\})W|\Omega\rangle}=K.
\end{equation*}
From \cref{lem:naimark} we have $W|\Omega\rangle=|\Omega\rangle\otimes |\xi_A\rangle \otimes |\xi_B\rangle \in H\otimes K_A\otimes K_B$. So, Alice's (resp.~Bob's) dilated operators act trivially on $K_B$ (resp.~$K_A$), and we have that
\begin{equation*}
K\subseteq (H\otimes K_A \otimes \C|\xi_B\rangle)\cap (H\otimes \C|\xi_A\rangle\otimes K_B)=WH    
\end{equation*}
Set $H'=W^*K\subseteq H$, and for $|v\rangle\in H'$, we observe that for any $a\in A$, $x\in X$, we have that $P_a^xW|v\rangle$ and $Q_b^yW|v\rangle$ are both in $K$, hence
\begin{equation}
    M_a^x|v\rangle=W^*P_a^xW|v\rangle\in H', \text{ and }N_b^y|v\rangle=W^*Q_b^yW|v\rangle\in H'.
\end{equation}
Additionally, $W|\Omega\rangle\in K$, so $|\Omega\rangle\in H'$, and $H'$ is invariant under every POVM effect. Indeed, since the original model $S$ is cyclic, we have that $H'=H$, and thus $K=WH$. It follows that $K$ is invariant under every dilated PVM measurement $\Pi$. That is, from $\Pi K\subseteq K$, we see that $W(W^*\Pi W)=\Pi W$ for all PVMs $\Pi$. In particular, letting $E=W^*\Pi W$ we see that $E^*=E$ is self-adjoint since $\Pi^*=\Pi$, and lastly the compression satisfies $E^2=(W^*\Pi W)^2=W^*\Pi^2 W=W^*\Pi W=E$ since $\Pi^2=\Pi$, and therefore every POVM effect $E$ is a projection. Since every original POVM in the model comes from such a compression, the model is projective.
\end{proof}

\subsection{Exact embezzlement of entanglement}

Let $|\varphi\rangle\in H_A\otimes H_B$. That is, we consider any finite-dimensional non-product (pure) quantum state as the ``target'' of our embezzlement protocol. The setup is as follows: there are two, spatially separated parties, each with their own local Hilbert spaces $H_A$ and $H_B$. They also have access to a third resource space $\widehat{H}$, on which they can share a predetermined state. Starting from local registers in the state $|e_0\rangle$ in their local spaces $H_A$ and $H_B$, their task is to create a target state on the joint register $H_A\otimes H_B$, by applying only local unitary operations $V_A$ on $H_A\otimes \widehat{H}$ and $V_B$ on $\widehat{H}\otimes H_B$ respectively.
 
 It is well-known that exact embezzlement is impossible when the spaces are all finite-dimensional \cite{van2003universal}. However, if we set $H_A=H_B$ to be finite-dimensional, and allow the resource space $\widehat{H}$ to be infinite-dimensional, and only demand that the operations $V_A\otimes \Id_{H_B}$ and $\Id_{H_A}\otimes V_B$ commute, then it was shown in \cite{cleve2017perfect} that exact embezzlement protocols exist for any state on $H_A\otimes H_B$.

 This commuting operator framework for embezzlement protocols is the one we will be using throughout the rest of the work. From now on, assume $d\geq 2$, and let $|\varphi\rangle=\sum_{i=0}^{d-1} \alpha_{i}|e_i\rangle \otimes |e_i\rangle$ such that $\sum_{i=0}^{d-1} \alpha_i^2=1$ in $\C^d\otimes \C^d$, be the (real) Schmidt decomposition of $|\varphi\rangle$.

\begin{definition}\label{def:con_embezzlement}
An \emph{exact embezzlement of entanglement protocol for $|\varphi\rangle \in \C^d\otimes \C^d$}, is a pair of unitaries $V_A\in \mbB(\C^d\otimes \widehat{H})$ and $V_B\in \mbB(\widehat{H}\otimes \C^d)$, and a catalyst state $|\Psi\rangle \in \widehat{H}$ such that $[V_A\otimes \Id_{d},\Id_{d}\otimes V_B]=0$ and 

\begin{equation}\label{eqn:embezz_EPR}
(V_A\otimes \Id_d)(\Id_d\otimes V_B)\left(|e_0\rangle \otimes |\Psi\rangle\otimes |e_0\rangle\right)=\sum_{i=0}^{d-1}\alpha_i |e_i\rangle \otimes |\Psi\rangle\otimes |e_i\rangle.
\end{equation}
Up to rearrangement this is the following transformation:
\begin{equation}
    |e_0\rangle\otimes |e_0\rangle \otimes |\Psi\rangle \mapsto |\varphi\rangle \otimes |\Psi\rangle.
\end{equation}
\end{definition}
Notably, the catalyst remains unchanged. Nonetheless, the choice of shared resource state is vital to enabling the embezzlement transformation. Indeed, $|\Psi\rangle \in \widehat{H}$ is the state from which the entanglement in $|\varphi\rangle$ is ``embezzled''. For example, when $d=2$ and $\alpha_{0}=\alpha_1=\frac{1}{\sqrt{2}}$ we have the embezzlement of an EPR pair. The work of \cite{cleve2017perfect} pioneered an operator algebraic way to view an exact embezzlement protocol using abstract states on $C^*$-algebras.

\begin{definition}\label{def:abs_embezzlement} 
    Let $\mcA$ and $\mcB$ be unital $C^*$-algebras and let $f$ be a state on $\mcA\otimes_{\max}\mcB$.  An \emph{exact embezzlement protocol for $|\varphi\rangle$ via the state $f$} consists of elements $\mcR=(r_0,\ldots, r_{d-1})\in \mcA^d$ and $\mcT=(t_0,\cdots,t_{d-1})\in \mcB^d$ satisfying
    \begin{equation}\label{eqn:embez_1}
        \sum_{i=0}^{d-1} r_i^*r_i\leq \Id_\mcA,\quad \text{and}\quad \sum_{i=0}^{d-1} t_i^*t_i\leq \Id_\mcB,
    \end{equation}
    and 
    \begin{equation}\label{eqn:embez_2}
        f(r_i\otimes t_j)=\delta_{ij}\alpha_{i}, \quad \text{ for $0\leq i, j\leq d-1$.}
    \end{equation}
    We say that \emph{$f$ exactly embezzles $\ket{\phi}$} if a pair $(\mcR,\mcT)$ exists for which \cref{eqn:embez_1} and \cref{eqn:embez_2} hold.
\end{definition}

\begin{proposition}\label{prop:equiv_def}
    A unit vector $|\varphi\rangle\in\C^d\otimes\C^d$ admits an exact embezzlement protocol in the operator formulation (\cref{def:con_embezzlement}) if and only if it admits one in the abstract-state formulation (\cref{def:abs_embezzlement}).
\end{proposition}

\begin{proof}
         To see this, let $(H_f,\pi_f,|\Omega_f\rangle)$ be the GNS representation of $f$. By Cauchy-Schwarz we have
         \begin{equation*}
             1=\sum_{i=0}^{d-1} \alpha_{i}^2\leq \sum_{i,j=0}^{d-1}\|\pi_f(r_i\otimes t_j)|\Omega_f\rangle\|^2=f\left((\sum_{i=0}^{d-1} r_i^*r_i)\otimes (\sum_{i=0}^{d-1} t_i^*t_i)\right)\leq 1,
         \end{equation*}
         from which it follows that $\pi_f(r_i\otimes t_j)|\Omega_f\rangle=\delta_{ij}\alpha_{i}|\Omega_f\rangle$ for $0\leq i,j\leq d-1$.
         
         Now, define $R=\sum_{i,k}E_{ik}\otimes R_{ik}$, and $T=\sum_{j,\ell}T_{j\ell}\otimes E_{j\ell}$, with $R_{i0}=\pi_f(r_i)$ and $T_{j0}=\pi_f(t_j)$ for $0\leq i,j\leq d-1$, setting the other columns to zero. It follows that
\begin{align*}
        (R\otimes \Id_d)(\Id_d\otimes T)\left(|e_0\rangle \otimes |\Omega_f\rangle\otimes |e_0\rangle\right)&=\sum_{i,j=0}^{d-1} |e_i\rangle \otimes \pi_f(r_i\otimes t_j)|\Omega_f\rangle\otimes |e_j\rangle\\
        &=\sum_{i,j=0}^{d-1}\alpha_i |e_i\rangle \otimes |\Omega_f\rangle\otimes \delta_{ij}|e_j\rangle\\
        &=\sum_{i=0}^{d-1}\alpha_i |e_i\rangle \otimes |\Omega_f\rangle\otimes |e_i\rangle.
\end{align*}
After rearrangement, this gives the embezzlement transformation
\begin{equation}
    |e_0\rangle\otimes |e_0\rangle \otimes |\Omega_f\rangle \mapsto|\varphi\rangle \otimes |\Omega_f\rangle,
\end{equation}
as desired. Finally, one can check that the operators $(R\otimes \Id_d)$ and $(\Id_d\otimes T)$ commute.

The observant reader will note that $R$ and $T$ are only required to be contractions. To obtain unitaries, one can start with the contractions and employ the Halmos dilation trick to obtain commuting unitaries \cite[Corollary 2.15]{harris2026self}, fully recovering the conventional triple $(V_A,V_B,|\Psi\rangle)$ required for an exact embezzlement of entanglement protocol.

For the other direction, one can check that taking $\mcR=(R_0,\ldots, R_{d-1})$ (resp. $\mcT=(T_0,\ldots, T_{d-1})$) to be the first column of the matrix $V_A\otimes \Id_d$ (resp. $\Id_\otimes V_B$) along with the state $|\Psi\rangle$ gives an abstract state embezzlement protocol. Indeed, the state $f(r_i\otimes t_j)=\langle\Psi|R_iT_j|\Psi\rangle$ on $C^*(\mcR)\otimes_{\max}C^*(\mcT)$ satisfies $\langle \Psi|R_iT_j|\Psi\rangle=\delta_{ij}\alpha_{i}$ for $0\leq i,j\leq d-1$. Now, since $V_A$ (resp.~$V_B$) are unitary we have $\sum_i R_i^*R_i=\Id$ (resp.~$\sum_j T_j^*T_j=\Id$), completing the proof.
\end{proof}

 Several works have built upon exact embezzlement of entanglement in the commuting operator framework, such as \cite{liu2025embezzlement,liu2025explicit,van2025multipartite}, and recently \cite{harris2026self}. The work of Harris made a key observation regarding the properties of the abstract state $f$. In particular, Harris examined the structure of the $C^*$-algebras $\mcA$ and $\mcB$ in the protocol and observed that they have a more succinct algebraic description. Recall the Cuntz algebra from the preliminaries. 

\begin{theorem}[\cite{harris2026self}, Theorem 3.4]\label{thm:unique_state_embezzlement}
    There is a unique state $f:\mcO_d\otimes_{\max} \mcO_d\to \C$ such that $f(v_i\otimes w_j)=\delta_{ij}\alpha_i$ for each $i,j=0,\ldots,d-1$.
 \end{theorem}

\section{The game that self-tests embezzlement}

The game we use for the commuting operator self-test is structurally the same as in \cite{coladangelo2020two} (and the follow-up \cite{beigi2021separation}), but we change one of the subgames involved. For self-testing the $3$-dimensional maximally entangled state, the $3$-CHSH game and the SATWAP inequality are used in \cite{coladangelo2020two} and \cite{beigi2021separation}, respectively. Our replacement will be Feige's game \cite{Feige1991} which has the advantage that it is a genuine nonlocal game with only two questions and three answers per player. 

\subsection{Feige's game}
We now describe Feige's game \cite{Feige1991}. In this game, Alice and Bob both receive a bit as question (with uniform probability), and they are each allowed to choose an answer from the set $\{0,1,\bot\}$. To win the game, exactly one of the players has to answer with $\bot$, while the second player has to guess the input of the first player, i.e.~the answer is a bit that is equal to the question of the first player. 

\begin{definition}
    Feige's game $G_F$  is defined via
$$
    X_F= Y_F=\{0,1 \}, \quad A_F=B_F=\{0,1,\bot \}, \quad \pi(x,y)=\tfrac{1}{4} \quad  \forall (x,y)
$$
and has the predicate
$$
    V_F(a,b,x,y)=\begin{cases}
        1 \quad (a,b)=(\bot,x) \text{ or } (a,b)=(y,\bot),
        \\
        0 \quad \text{otherwise.}
    \end{cases}
$$
\end{definition}

\begin{proposition}[\cite{schmidt2025guess}]\label{prop:feige_value}
Feige's game has quantum advantage with $\omega_{qc}(G_F)=\frac{9}{16}>\frac{1}{2}=\omega_c(G_F)$.
\end{proposition}

To prove \cref{prop:feige_value} the authors of \cite{schmidt2025guess} find an SOS decomposition. From the SOS decomposition one can determine much more than quantum advantage.

\begin{theorem}[\cite{schmidt2025guess}]\label{thm:Feige}
Consider the nonlocal game $G_F$ defined above: 
\begin{enumerate}
    \item The optimal commuting operator value $\omega_{qc}(G_F)=9/16$ is achieved by a unique optimal correlation $p$,
    \item $p$ is a commuting operator self-test,
    \item in the cyclic restriction of any optimal strategy, Alice and Bob's measurement operators generate local copies of $M_3(\C)$, and
    \item in the optimal GNS representation:
    \begin{align*}
      E_{ij}^A\ket{\Omega}=E_{ji}^B\ket{\Omega}, \quad \bra{\Omega} E_{ij}^A\ket{\Omega}=\bra{\Omega} E_{ij}^B\ket{\Omega}=\delta_{ij}\frac{1}{3},\; \text{ and }\;  \pi(p_0^{(0)})\ket{\Omega}=E_{00}^A\ket{\Omega},
    \end{align*}
    where $E_{ij}$ $i,j\in \{0,1,2\}$ are the $3\times 3$ matrix units.
\end{enumerate}
\end{theorem}

\begin{proof}
    \cite[Theorem 4.6]{schmidt2025guess} establishes that the optimal correlation for $G_F$ is a commuting operator self-test among projective models, since its SOS and determining relations extend to commuting representations. However, since \cite[Lemma 4.5]{schmidt2025guess} proves that the local algebras in any optimal GNS representation are isomorphic to $M_3(\C)$ we can conclude by \cref{lem:projective_to_POVM} that it is a commuting operator self-test among all POVM models. Moreover, the proof of \cite[Lemma 4.5]{schmidt2025guess} allows us to concretely establish \emph{(4)}, by considering the properties of the ideal (reduced) model. From this, we observe that $\pi(p_0^{(0,F)})\ket{\Omega}=E_{00}^A\ket{\Omega}$, establishing the final part of \emph{(4)}.
\end{proof}

\subsection{Tilted CHSH game}

Now, we look at the tilted CHSH game and define the variant introduced in \cite{coladangelo2020two}. This variant is equivalent to the original tilted CHSH game and is obtained as follows: take the original predicate of the tilted CHSH game, flip the answer labels of both players and then swap the roles of Alice and Bob.
\begin{definition}
For $\alpha\in (0,1]$, the (variant of the) tilted CHSH game $G_{tCHSH(\alpha)}$ is defined via
$$
    X_{t_{CHSH}}= Y_{t_{CHSH}}=A_{t_{CHSH}}=B_{t_{CHSH}}=\{0,1 \}, \quad \pi(x,y)=\tfrac{1}{4} \quad  \forall (x,y)
$$
and the predicate
$$
    V_{tCHSH(\alpha)}(a,b|x,y)=(-1)^{a\oplus b- x y}-\delta_{x=y=0}\beta(-1)^b.
    $$
Here $\alpha$ and $\beta\in [0,2)$ are related as follows: There is an angle $\phi\in (0,\frac{\pi}{4}]$ such that $\sin(2\phi)=\frac{\sqrt{4-\beta^2}}{\sqrt{4+\beta^2}}$ and $\alpha=\tan(\phi)$.
\end{definition}

\begin{proposition}[\cite{bamps2015sum}]\label{prop:tilted-score}
The tilted CHSH game $G_{tCHSH(\alpha)}$ has quantum advantage and it holds $\omega_{qc}(G_{tCHSH(\alpha)})=\frac{1}{4}\sqrt{8+2\beta^2}$.
\end{proposition}

For the purposes of this work, we are interested in the case $\alpha=\frac{1}{\sqrt{2}}$, which yields $\beta=\frac{2}{\sqrt{17}}$.

\begin{theorem}\label{thm:tCHSH}
Consider the generalized nonlocal game $G_{tCHSH(\frac{1}{\sqrt{2}})}$ from above: 
    \begin{enumerate}
        \item Its optimal (quantum) commuting operator value of $\tfrac{3}{\sqrt{17}}$ is achieved by a unique optimal correlation $p$, 
        \item $p$ is a commuting operator self-test,
        \item
        in the cyclic restriction of any optimal strategy, Alice and Bob's measurement operators generate local copies of $M_2(\C)$, and
        \item in any optimal GNS representation we have $$E_{00}^A|\Omega\rangle=E_{00}^B|\Omega\rangle,\; E_{10}^B|\Omega\rangle=\tfrac{1}{\sqrt{2}}E_{01}^A|\Omega\rangle,\; E_{10}^A|\Omega\rangle=\tfrac{1}{\sqrt{2}}E_{01}^B|\Omega\rangle\; \text{ and } \; \pi(q_0^{(0)})\ket{\Omega}=E_{00}^B\ket{\Omega},$$
        where $E_{ij}$, for $i,j\in \{0,1\}$, are the $2\times 2$ matrix units.
    \end{enumerate}
\end{theorem}

It was stated in \cite[Proposition 7.8]{PSZZ} that the tilted CHSH correlations for all $\beta\in [0,2)$ are commuting operator self-tests for POVM quantum (commuting operator) models. However, since the publication the authors have noted an error in part of the proof. Although we believe the statement of \cite[Proposition 7.8]{PSZZ} is still correct, we have opted here to provide an alternate proof for our specific case. Both approaches are based on the SOS decompositions of Bamps and Pironio  \cite{bamps2015sum} to determine uniqueness of the state.

Before we give the proof we recall the following well-known result:

\begin{lemma}\label{lem:anti_pair}
    If $X_i,Z_i$ for $i=1,2$ are two anticommuting pairs of self-adjoint unitaries that commute with each other, i.e. $X_i=X_i^*$, $Z_i=Z_i^*$, $X_i^2=Z_i^2=\Id$ for $i=1,2$, $[X_1,X_2]=[Z_1,Z_2]=[X_1,Z_2]=[X_2,Z_1]=0$, and $X_iZ_i+Z_iX_i=0$ for $i=1,2$, that act irreducibly on a Hilbert space $H$, then there is a unitary $U:H\to \C^2\otimes \C^2$, so that $X_1\cong \sigma_x\otimes \Id$, $Z_1\cong \sigma_Z\otimes \Id$, and $X_2\cong \Id \otimes \sigma_X$, $Z_2\cong\Id \otimes \sigma_Z$.
\end{lemma}

\begin{proof}[Proof of \cref{thm:tCHSH}]
    Let $a_x=p_0^x-p_1^x$ and $b_y=q_0^y-q_1^y$ be the canonical self-adjoint unitaries in $\mcA_{\mathrm{PVM}}^{(2,2)}\otimes_{\max} \mcA_{\mathrm{PVM}}^{(2,2)}$, with $[a_x,b_y]=0$ for all $x,y\in \{0,1\}$. Fix $\alpha=\frac{1}{\sqrt{2}}$ so that $\beta=\frac{2}{\sqrt{17}}$ as above, and define the elements $T=(a_0+a_1)$ and $S=(a_0-a_1)$ for convenience.
    Then, the unnormalized Bell element can be written as
    \begin{equation*}
        \eta=-\beta b_0+b_0T+b_1S,    
    \end{equation*}
     and by \cref{prop:tilted-score} the optimal (unnormalized) commuting operator score is $\lambda:=\sqrt{8+2\beta^2}=\frac{12}{\sqrt{17}}$.
    
    In \cite{bamps2015sum} a pair of noncommutative SOSs were derived with respect to the abstract PVM elements, establishing that
    \begin{equation}
        2\lambda(\lambda-\eta)=r_1^2+r_2^2=r_3^2+r_4^2,
    \end{equation}
    where
    \begin{align*}
        r_1&=\lambda-\eta,\\
        r_2&=-\beta b_1 -b_0 S- b_1 T,\\
        r_3&=2b_0-\tfrac{\lambda}{2}T-\tfrac{\beta}{2}(b_0T-b_1S),\text{ and}\\
        r_4&=2b_1-\tfrac{\lambda}{2}S-\tfrac{\beta}{2}(b_0S-b_1T).
    \end{align*}
If $f$ is a pure state on $\mcA_{\mathrm{PVM}}^{(2,2)}\otimes_{\max} \mcA_{\mathrm{PVM}}^{(2,2)}$ such that $f(\eta)=\lambda$, then in the GNS we have $\pi_f(r_i)|\Omega_f\rangle=0$ for all $i=1,2,3,4$. Letting $c_0=\tfrac{\sqrt{17}}{6}(a_0+a_1)$, and $c_1=\tfrac{\sqrt{17}}{4\sqrt{2}}(a_0-a_1)$, one can determine that
\begin{equation*}
    b_0-c_0=\frac{-\sqrt{17}r_1+17r_3-\frac{\sqrt{17}}{3}b_1(r_2-\sqrt{17}r_4)}{32},
\end{equation*}
and
\begin{equation*}
    c_1-\frac{1}{2\sqrt{2}}b_1(3\cdot 1_\mcM+c_0)=\frac{\sqrt{17}}{24\sqrt{2}}(r_2-\sqrt{17}r_4).
\end{equation*}
Hence,
\begin{equation}\label{eqn:transfer_0}
    \pi_f(c_0)|\Omega_f\rangle=\pi_f(b_0)|\Omega_f\rangle,
\end{equation}
and
\begin{equation}\label{eqn:transfer_1}
    \pi_f(c_1)|\Omega_f\rangle=\frac{1}{2\sqrt{2}}\pi_f(b_1)(3\cdot \Id+\pi_f(b_0))|\Omega_f\rangle.
\end{equation}
In the cyclic Hilbert space $H_f$ let $\pi_f(a_x)=A_x$, for $x=0,1$, $\pi_f(b_y)=B_y$ and $\pi_f(c_y)=C_y$, for $y=0,1$.

From the above equations, we can deduce that

\begin{equation*}
C_0^2|\Omega_f\rangle=C_0B_0|\Omega_f\rangle=B_0C_0|\Omega_f\rangle=B_0^2|\Omega_f\rangle=|\Omega_f\rangle.
\end{equation*}
Additionally, we observe that
\begin{equation*}
    C_0^2=\frac{17}{36}(2\Id+A_0A_1+A_1A_0).
\end{equation*}
In particular, $C_0^2$ commutes with not only Bob's operators, but also Alice's operators. Then, because $f$ is pure, $\pi_f$ is irreducible, and Schur's lemma implies that any operator commuting with the whole representation is scalar, hence $C_0^2=\Id$.

On the other hand, from the observation that $$(A_0+A_1)^2+(A_0-A_1)^2=4\Id,$$ we obtain $${36}C_0^2+{32}C_1^2=68\Id.$$
Substituting in $C_0^2=\Id$, we deduce that $$C_1^2=\frac{(68-36)}{32}\Id=\Id,$$
as desired.

Finally, we see that
\begin{equation*}
C_0C_1+C_1C_0=\frac{17}{24\sqrt{2}}\left((A_0+A_1)(A_0-A_1)+(A_0-A_1)(A_0+A_1)\right)=0.
\end{equation*}
From the above equations we see that
\begin{equation*}
    \frac{1}{2\sqrt{2}}(B_0B_1+B_1B_0)(3\Id+C_0)|\Omega_f\rangle=0.
\end{equation*}
Since $(3\Id+C_0)$ is invertible and commutes with $B_0$ and $B_1$, we see that
\begin{equation*}
    (B_0B_1+B_1B_0)|\Omega_f\rangle=0.
\end{equation*}
Again, because the representation is irreducible, we can use Schur's lemma to deduce that 
\begin{equation}\label{eqn:bob_anticommutator}
    B_0B_1+B_1B_0=0,
\end{equation}
since the anticommutator in \cref{eqn:bob_anticommutator} commutes with the generators of both local algebras. Hence, by \cref{lem:anti_pair}, up to the unitary identification $U:H_f\to \C_2\otimes \C_2$, we can write $UC_0U^*= \sigma_Z\otimes \Id_2:=\wtd{C}_0$, $UC_1U^*= \sigma_X\otimes \Id_2:=\wtd{C}_1$, and $UB_0U^*= \Id_2 \otimes \sigma_Z:=\wtd{B}_0$, $UB_1U^*= \Id_2 \otimes \sigma_X:=\wtd{B}_1$. Solving for Alice's observables we see that
\begin{equation*}
    UA_0U^*=\frac{3\sigma_Z+2\sqrt{2}\sigma_X}{\sqrt{17}}\otimes \Id_2:=\wtd{A}_0,\quad UA_1U^*=\frac{3\sigma_Z-2\sqrt{2}\sigma_X}{\sqrt{17}}\otimes \Id_2:=\wtd{A}_1.
\end{equation*}
For the state, let $|\wtd{\psi}\rangle=U|\Omega_f\rangle$. Plugging these operators into \cref{eqn:transfer_0} and \cref{eqn:transfer_1} above, we obtain
\begin{equation*}
    (\sigma_Z\otimes \Id)|\wtd{\psi}\rangle=(\Id\otimes \sigma_Z)|\wtd{\psi}\rangle,
\end{equation*}
and
\begin{equation*}
    (\sigma_X\otimes \Id)|\wtd{\psi}\rangle=\tfrac{1}{2\sqrt{2}}(\Id\otimes \sigma_X(3\Id+\sigma_Z))|\wtd{\psi}\rangle,
\end{equation*}
The first equation implies that $|\wtd{\psi}\rangle=u|00\rangle+v|11\rangle$, while the second implies that $v=\sqrt{2}u$. Since $|u|^2+|v|^2=1$ we see that (up to a complex phase) $u=\frac{1}{\sqrt{3}}$ and $v=\frac{\sqrt{2}}{\sqrt{3}}$, so that

\begin{equation*}
    |\wtd{\psi}\rangle=\frac{|00\rangle+\sqrt{2}|11\rangle}{\sqrt{3}}
\end{equation*}

One can verify that this commuting operator model consisting of $$(\C^2\otimes\C^2,\wtd{A_0},\wtd{A_1},\wtd{B_0},\wtd{B_1},|\wtd{\psi}\rangle)$$ attains the value $\frac{3}{\sqrt{17}}$ in $G_{tCHSH(\frac{1}{\sqrt{2}})}$. Hence, by \cref{prop:optimal_state} we conclude that $f$ is the unique optimal projective state obtaining $\frac{3}{\sqrt{17}}$ for the game $G_{tCHSH(\frac{1}{\sqrt{2}})}$. The ideal model $S_f$ is locally cyclic, with $\pi_f(\mcA_\mathrm{PVM}^{(2,2)}\otimes \mcI)\cong M_2(\C)\otimes I_2$ and $\pi_f(\mcI\otimes \mcA_\mathrm{PVM}^{(2,2)})\cong I_2\otimes M_2(\C)$ as local measurement algebras. Hence, by \cref{lem:projective_to_POVM}, every optimal POVM model is projective, establishing uniqueness among all commuting operator models.

Now, we can take

\begin{align*}
    E_{00}^A&=\frac{\Id+C_0}{2}, & E_{11}^A=&\frac{\Id-C_0}{2}, & E_{01}^A&=E_{00}^AC_1, & E_{10}^A&=C_1E_{00}^A\\
    E_{00}^B&=\frac{\Id+B_0}{2}, & E_{11}^B=&\frac{\Id-B_0}{2}, & E_{01}^B&=E_{00}^BB_1, & E_{10}^B&=B_1E_{00}^B.
\end{align*}
Then the unitary identification $UE_{ij}^AU^*= e_{ij}\otimes \Id_2$, and $UE_{ij}^BU^*= \Id_2\otimes e_{ij}$, gives

$$UE_{00}^A|\Omega\rangle=(|0\rangle\langle0|\otimes \Id_2)|\wtd{\psi}\rangle=\frac{1}{\sqrt{3}}|0\rangle|0\rangle=(\Id_2\otimes |0\rangle\langle0|)|\wtd{\psi}\rangle=UE_{00}^B|\Omega\rangle,$$

and $$UE_{01}^A|\Omega\rangle=(|0\rangle\langle1|\otimes \Id_2)|\wtd{\psi}\rangle=\sqrt{\frac{2}{3}}|0\rangle|1\rangle=\sqrt{2}(\Id_2\otimes |1\rangle\langle0|)|\wtd{\psi}\rangle=\sqrt{2} U E_{10}^B|\Omega\rangle.$$

To complete the proof we can take $\wtd{B}_0=\Id_2\otimes \sigma_Z$, and therefore $\widetilde{Q}_0^0=\frac{\tilde{B}_0+\Id}{2}=\Id_2\otimes e_{00}$, as desired.
\end{proof}

\subsection{Coladangelo's embezzlement game}

Our version of Coladangelo's game consists of the two games presented in the previous section together with a third test that also sends questions from Feige's game to Alice and the tilted CHSH game to Bob. Here is the precise description. 

\begin{definition}\label{def:embgame}
    The game $G_{emb}$ has question and answer sets 
    \begin{align*}
       X_{emb}&=Y_{emb}=\{(0,F),(1,F),(0,t_{CHSH}), (1,t_{CHSH})\},\\ A_{emb}&=\{0,1,\bot\},B_{emb}=\{0,1,\bot\}.   
    \end{align*}
   It consists of three parts, each occurring with equal probability.
    \begin{itemize}
        \item[(a)] The verifier sends both players questions $(x,F), (y,F)$ from Feige's game,
and they win according to the predicate $V_F$.
        \item[(b)] The verifier sends both players questions $(x,t_{CHSH}), (y,t_{CHSH})$ from the tilted
CHSH game. They obtain a score according to the function 
\begin{align*}
    \tilde V_{tCHSH(\frac{1}{\sqrt{2}})}(a,b|x,y)=\begin{cases}
        V_{tCHSH(\frac{1}{\sqrt{2}})}(a,b|x,y) \text{ if } a,b\in\{0,1\}\\
        -2 \text{ if }a=\perp \text{ or }b=\perp.
    \end{cases}
\end{align*}
 Note that since $-2 < -1-\frac{2}{\sqrt{17}}$, the players will increase their score if they replace each time they answer $\perp$ by $1$. Therefore, no player answers $\perp$ on tilted-CHSH questions in an optimal strategy.
        \item[(c)] Alice receives the question $(0,F)$ from Feige's game, and Bob receives $(0,t_{CHSH})$ from the game $G_{tCHSH(\frac{1}{\sqrt{2}})}$. They win if both players answer $0$ or neither player answers $0$. Otherwise they lose.\label{def:test3}
    \end{itemize}
\end{definition}

\begin{proposition}\label{prop:value}
    The game $G_{emb}$ has quantum advantage with 
    \begin{align*}
        \omega_{qc}(G_{emb})=\frac{1}{3}(\omega_{qc}(G_F)+\omega_{qc}(G_{tCHSH(\frac{1}{\sqrt{2}})})+1)=\frac{3}{16}+\frac{1}{\sqrt{17}}+\frac{1}{3}=\frac{25}{48}+\frac{1}{\sqrt{17}}.
    \end{align*}
\end{proposition}

\begin{proof}
The argument for the upper bound is the same as in [\cite{coladangelo2020two}, Proposition 2]. 
To show that the upper bound is achieved, consider the embezzling family 
\begin{align*}
    \ket{\Gamma_d}=\frac{1}{\sqrt{N_d}}\sum_{k=1}^{2d} \frac{\ket{k}\otimes \ket{k}}{\sqrt{k}}\in \C^{2d}\otimes \C^{2d}, \quad N_d=\sum_{k=1}^{2d} \frac{1}{k}
\end{align*}
of \cite{van2003universal}. Define $U_d:\C^3\otimes \C^{2d}\to \C^3\otimes \C^{2d}$ such that 
\begin{align*}
    U_d(\ket{0}\otimes \ket{k})&=\ket{0}\otimes \ket{k}, \\ U_d(\ket{1}\otimes \ket{2j-1})&=\ket{1}\otimes \ket{j},\quad &&U_d(\ket{1}\otimes \ket{2j})=\ket{2}\otimes \ket{j},\\
    U_d(\ket{2}\otimes \ket{2j-1})&=\ket{1}\otimes \ket{d+j},\quad &&U_d(\ket{2}\otimes \ket{2j})=\ket{2}\otimes \ket{d+j},
\end{align*}
for $1\leq k\leq 2d$, $1\leq j\leq d$, which is a unitary since it maps basis vectors bijectively. 
Now consider $\ket{\Psi_2}=\frac{1}{\sqrt{3}}(\ket{00}+\sqrt{2}\ket{11})\in \C^3\otimes \C^3$ and $\ket{\Phi_3}=\frac{1}{\sqrt{3}}(\ket{00}+\ket{11}+\ket{22})\in \C^3\otimes \C^3$, the states in the optimal quantum strategies for the tilted CHSH game and Feige's game, respectively. Then it holds 
\begin{align}
    \|(U_d\otimes U_d)(\ket{\Psi_2}\otimes \ket{\Gamma_d})-\ket{\Phi_3}\otimes \ket{\Gamma_d}\|= O((\log{d})^{-\frac{1}{2}}). \label{eq:emb}
\end{align}
Therefore, we obtain a family of quantum strategies converging to the upper bound as follows. Alice and Bob share the entangled state $\ket{\psi_d}=\ket{\Psi_2}_{AB}\otimes \ket{\Gamma_d}_{A'B'}$. They use the following operators.
\begin{itemize}
    \item For the tilted CHSH, consider the ideal qubit measurements $\tilde E$ on $\mathrm{span} \{\ket{0}, \ket{1}\}$ and put $\tilde P^{(x,t_{CHSH})}_{\perp}=\tilde Q^{(y,t_{CHSH})}_{\perp}=\ket{2}\bra{2}$. Note that $\tilde Q^{(0,t_{CHSH})}_0=\ket{0}\bra{0}$. In the strategy, the  measurements $E=\tilde E\otimes I_{2d}$ are used.
    \item For Feige's game, choose the optimal qutrit measurements $\tilde F$ with Alice’s projection equal to $\tilde P_0^{(0,F)}=\ket{0}\bra{0}$. This is possible by a basis change, with the conjugate change on Bob that preserves $\ket{\Phi_3}$. In the strategy, the players use $F=U_{d}^{*}(\tilde F\otimes I_{2d})U_{d}$.
\end{itemize}
This strategy achieves score exactly $\frac{3}{\sqrt{17}}$ for tilted CHSH, the value for Feige's game converges to $\frac{9}{16}$ by \Cref{eq:emb} and the consistency test \emph{(c)} in \Cref{def:embgame} has exactly value $1$ because $(P_0^{(0,F)}\otimes I_{BB'})\ket{\psi_d}=(I_{AA'}\otimes Q^{(0,t_{CHSH})}_0)\ket{\psi_d}=\frac{1}{\sqrt{3}}\ket{00}_{AB}\otimes \ket{\Gamma_d}_{A'B'}$, since $U_d$ fixes $\ket{0}\otimes \ket{k}$.
\end{proof}

\section{Uniqueness of the state}

The goal of this section is to give a proof of our main result (\cref{thm:comopselftest}). We start with a lemma that explains the impact of subgame \emph{(c)} in \Cref{def:embgame}.

\begin{lemma}\label{lem:overlap}
Let $S=(\ket{\psi}\in H, M_a^{\tilde x}, N_b^{\tilde y})$ be an optimal commuting operator strategy for the game $G_{emb}$. Then
\begin{align*}
    M_0^{(0,F)}\ket{\psi}= N_0^{(0,t_{CHSH})}\ket{\psi}.
\end{align*}
\end{lemma}

\begin{proof}
Since the strategy is optimal, it achieves value $1$ for test $(c)$ in \Cref{def:test3}. This means that the probability of exactly one of them answering $0$ given the questions for test $(c)$ is zero, i.e. 
\begin{align*}
  \bra{\psi}M^{(0,F)}_0N^{(0,t_{CHSH})}_t\ket{\psi}=\bra{\psi}M^{(0,F)}_sN^{(0,t_{CHSH})}_0\ket{\psi}=0  
\end{align*} 
for $s,t \neq 0$. Therefore, using $\sum_aM^{(0,F)}_a=1=\sum_b N^{(0,t_{CHSH})}_b$, we obtain
\begin{align*}
  \bra{\psi}M^{(0,F)}_0\ket{\psi} =\bra{\psi}M^{(0,F)}_0N^{(0,t_{CHSH})}_0\ket{\psi} =\bra{\psi}N^{(0,t_{CHSH})}_0\ket{\psi}.
\end{align*} 
We deduce 
\begin{align*}
    \|&M_0^{(0,F)}\ket{\psi}-N_0^{(0,t_{CHSH})}\ket{\psi}\|^2\\
    &=\bra{\psi}(M^{(0,F)}_0)^2\ket{\psi} -2\bra{\psi}M^{(0,F)}_0N^{(0,t_{CHSH})}_0\ket{\psi} +\bra{\psi}(N^{(0,t_{CHSH})}_0)^2\ket{\psi}\\
    &\leq\bra{\psi}M^{(0,F)}_0\ket{\psi} -2\bra{\psi}M^{(0,F)}_0N^{(0,t_{CHSH})}_0\ket{\psi} +\bra{\psi}N^{(0,t_{CHSH})}_0\ket{\psi}\\
    &=0,
\end{align*}
where we used the equality above and the fact that $M^{(0,F)}_0$ and $N^{(0,t_{CHSH})}_0$ are contractions. This finishes the proof.
\end{proof}

Now, we want to show that there is a unique state for the optimal correlation of $G_{emb}$. We will do this by inspecting \Cref{thm:Feige} and \Cref{thm:tCHSH} as well as the previous lemma, which tells us that the operators of an optimal strategy generate an algebra of the following form. Let $\mathcal{A}\subseteq \mbB(H)$ for $H\neq \{0\}$ be a unital algebra 
\begin{align}
    \mathcal{A}=C^*(e_{ij}, f_{ab}, i,j \in \{0,1,2\}, a,b\in&\{0,1\}|e_{ij}^* = e_{ji}, e_{ij} e_{kl} = \delta_{jk}e_{il}, \sum_ie_{ii}=1,\nonumber\\f_{ab}^* = &f_{ba}, f_{ab} f_{cd} = \delta_{bc}f_{ad},f_{00}=e_{00},f_{11}=e_{11}+e_{22}).\label{eq:gamealg}
\end{align}
Note that the generators $e_{ij}$ and $f_{ab}$ can be thought of as matrix units of $M_3(\C)$ and $M_2(\C)$, respectively. These matrix units have the additional relations $f_{00}=e_{00}, f_{11}=e_{11}+e_{22}$. This algebra contains the Cuntz algebra, as we see in the next lemma. 

\begin{lemma}\label{lem:isoCuntz}
    Consider the algebra $\mathcal{A}$ as in \Cref{eq:gamealg} and let $P=e_{00}=f_{00}$. Then there is a unital $*$-isomorphism $\theta:\mathcal O_2\mapsto P\mathcal{A}P$ with $V_r\mapsto f_{01}e_{r0}$.
\end{lemma}

\begin{proof}
First note that $P$ is the identity in the algebra $P\mathcal{A}P$. Now, define $\tilde V_r=f_{01}e_{r0}=f_{00}f_{01}e_{r0}e_{00}=Pf_{01}e_{r0}P\in P\mathcal{A}P$, $r\in \{1,2\}$. We compute
\begin{align*}
    \tilde V_r^*\tilde V_s=e_{0r}f_{10}f_{01}e_{s0}=e_{0r}f_{11}e_{s0}=e_{0r}(1-e_{00})e_{s0}=\delta_{rs}P
\end{align*}
and
\begin{align*}
    \tilde V_1 \tilde V_1^*+\tilde V_2 \tilde V_2^*&=f_{01}e_{10}e_{01}f_{10}+f_{01}e_{20}e_{02}f_{10}\\
    &=f_{01}(e_{11}+e_{22})f_{10}=f_{01}f_{11}f_{10}=P.
\end{align*}
By the universal property there exists a unital $*$-homomorphism $ \theta:\mathcal O_2\to P\mathcal AP$ such that $V_r \mapsto \tilde V_r=f_{01}e_{r0}$, $1\mapsto P$. 

The map $\theta$ is surjective because of the following. Consider 
\begin{align*}
    \mathcal B=\overline{\mathrm{span}}\{e_{i0}xe_{0j}| x \in C^*_{P\mcA P}(\tilde V_1, \tilde V_2), i,j \in \{0,1,2\}\}.
\end{align*}
This $C^*$-subalgebra of $\mathcal A$ contains all $e_{ij}$ as $P\in C^*_{P\mcA P}(\tilde V_1, \tilde V_2)$ and it contains $f_{01}$ because $f_{01}=\tilde V_1 e_{01}+\tilde V_2e_{02}$. Therefore $\mathcal B=\mathcal A$, which shows $P\mathcal AP=P\mathcal BP=C^*(\tilde V_1, \tilde V_2)$. 

Finally, the map $\theta$ is injective since the Cuntz algebra $\mathcal O_2$ is simple.
\end{proof}

Now, the idea is to use \Cref{thm:unique_state_embezzlement} for a unique state on the Cuntz algebra and then show that its extension is also unique. 

\begin{lemma}\label{lem:uniquestate}
    Fix noncommutative $*$-polynomials $P^x_a$ and $Q^y_b$ in symbols $\tilde e_{ij}, \tilde f_{uv}$, $i,j\in \{0,1, 2\}, u,v\in \{0,1\}$. Let $\phi:\mathcal A_{POVM}^{X,A}\otimes_{max}\mathcal A_{POVM}^{Y,B}\to \C$ be a state with GNS representation $(\pi_\phi,H_\phi,\ket{v_\phi})$ such that 
      $  \mathcal M_A=\pi_{\phi}(A_{POVM}^{X,A}\otimes 1), \mathcal M_B=\pi_{\phi}(1\otimes A_{POVM}^{Y,B})$
    are generated by elements $\tilde e_{ij}^A, \tilde f_{uv}^A$ and $\tilde e_{ij}^B, \tilde f_{uv}^B$, respectively, that each satisfy the relations in \Cref{eq:gamealg}. Assume that 
    \begin{align}
        \pi_{\phi}(m^x_a\otimes 1)=P^x_a(\tilde e^A_{ij},\tilde f^A_{uv}) \text{ and } \pi_{\phi}(1\otimes n^y_b)=Q^y_b(\tilde e^B_{ij},\tilde f^B_{uv}),\label{eq:localgebra}
    \end{align}
        \begin{align}
    \tilde e_{ij}^A\ket{v_\phi}=\tilde e_{ji}^B\ket{v_\phi},\label{eq:flip}
\end{align}
    \begin{align}
    \bra{v_{\phi}}\tilde f_{01}^A\tilde e_{r0}^A\tilde f_{01}^B\tilde e_{s0}^B\ket{v_{\phi}}=\frac{1}{3\sqrt{2}}\delta_{rs}, \quad r,s\in\{1,2\}.\label{eq:valuegen}
\end{align}
Then the state $\phi$ is unique.
\end{lemma}

\begin{proof}
We compute 
\begin{align}
    \bra{{v_\phi}}\tilde e_{ii}^A\ket{v_\phi}
    =\bra{{v_\phi}}\tilde e_{ij}^A\tilde e_{ji}^A\ket{v_\phi}=\bra{{v_\phi}}\tilde e_{ji}^B\tilde e_{ij}^B\ket{v_\phi}=\bra{{v_\phi}}\tilde e_{jj}^B\ket{v_\phi}.\label{eq:normalization}
\end{align}
Therefore $\bra{{v_\phi}}\tilde e_{ii}^A\ket{v_\phi}=\bra{{v_\phi}}\tilde e_{kk}^A\ket{v_\phi}$ for $i,k\in\{0,1,2\}$ and since we have $\phi(1)=1$ and $\sum_{a=0}^2\tilde e_{aa}=1$, we deduce $\bra{{v_\phi}}\tilde e_{ii}^A\ket{v_\phi}=\bra{{v_\phi}}\tilde e_{jj}^B\ket{v_\phi}=\frac{1}{3}$ for all $i,j\in\{0,1,2\}$. Let $P^A=\tilde e_{00}^A$ and $P^B=\tilde e_{00}^B$. For $x^A\in P^A\mathcal M_AP^A, y^B\in P^B \mathcal M_BP^B$, it holds
\begin{align}
    \bra{{v_\phi}}\tilde e_{i0}^Ax^A\tilde e_{0j}^A\tilde e_{k0}^B y^B\tilde e_{0l}^B\ket{v_\phi}&=\bra{{v_\phi}}\tilde e_{0i}^Bx^A\tilde e_{0j}^A\tilde e_{k0}^B y^B\tilde e_{l0}^A\ket{v_\phi}\nonumber\\
    &=\bra{{v_\phi}}x^A\tilde e_{0j}^A\tilde e_{l0}^A\tilde e_{0i}^B\tilde e_{k0}^B y^B\ket{v_\phi}\nonumber\\
&=\delta_{jl}\delta_{ik}\bra{{v_\phi}}x^A y^B\ket{v_\phi}\label{eq:statecomp}
\end{align}
where we used \Cref{eq:flip} several times. 

Define $\psi:\mathcal O_2\otimes_{\max} \mathcal O_2\to \C$ by
\begin{align*}
    \psi(x\otimes y)=3\bra{{v_\phi}}\theta_A(x)\theta_B(y)\ket{v_\phi},
\end{align*}
where $\theta_C:\mathcal O_2\to P^C\mathcal M_CP^C$ is the $*$-isomorphism from \Cref{lem:isoCuntz} for $C=A,B$. It holds
\begin{align*}
    \psi(1\otimes 1)=3\bra{{v_\phi}}\tilde e_{00}^A \tilde e_{00}^B\ket{v_\phi}=3\bra{{v_\phi}}\tilde e_{00}^A\ket{v_\phi}=1
\end{align*}
by \Cref{eq:normalization}  and $\psi$ is positive since it is a vector state for the vector $\sqrt{3}P^AP^B\ket{v_\phi}$. Therefore $\psi$ is a state. Furthermore
\begin{align*}
   \psi(V_r^A\otimes V_s^B)=3\bra{v_{\phi}}\tilde f_{01}^A\tilde e_{r0}^A\tilde f_{01}^B\tilde e_{s0}^B\ket{v_{\phi}}=\frac{1}{\sqrt{2}}\delta_{rs}
\end{align*}
by \Cref{eq:valuegen}. By \Cref{thm:unique_state_embezzlement} the state $\psi$ is unique.

Using \Cref{eq:statecomp}, we get
\begin{align*}
    \bra{{v_\phi}}\tilde e_{i0}^A\theta_A(x)\tilde e_{0j}^A\tilde e_{k0}^B \theta_B(y)\tilde e_{0l}^B\ket{v_\phi}=\delta_{jl}\delta_{ik}\bra{{v_\phi}}\theta_A(x) \theta_B(y)\ket{v_\phi}
    =\frac{\delta_{ik}\delta_{jl}}{3}\psi(x\otimes y)
\end{align*}
By \Cref{eq:localgebra} the measurement generators are fixed expressions in the matrix units. The elements $\tilde e_{i0}^A\theta_A(x)\tilde e_{0j}^A\tilde e_{k0}^B \theta_B(y)\tilde e_{0l}^B$ span a dense $*$-subalgebra of $\pi_{\phi}(\mathcal A_{POVM}^{X,A}\otimes_{max}\mathcal A_{POVM}^{Y,B})$ and since the polynomials in \Cref{eq:localgebra} are fixed, the matrix-unit and Cuntz relations express each word in the measurement generators as a linear combination of the displayed elements, with coefficients and words independent of the representation. Therefore \Cref{eq:statecomp} uniquely determines $\phi$ from values of $\psi$. Since $\psi$ is unique, there is at most one state $\phi$ satisfying these conditions.
\end{proof}

We are now ready to prove the main theorem.

\begin{theorem}\label{thm:comopselftest}
      The generalized nonlocal game $G_{emb}$ has a unique optimal commuting operator correlation $\hat{p}$. This correlation is a commuting operator self-test and belongs to $C_{qa}(4,4,3,3)\setminus C_{qs}(4,4,3,3)$.
\end{theorem}

Before we prove the theorem we recall a well-known lemma from operator algebras, with the argument in the proof being similar to [\cite{murphy}, p. 134].

\begin{lemma}\label{lem:separating}
Consider unital $C^*$-algebras $\mathcal A_1, \mathcal A_2$. Let $(\pi, H, \ket{v})$ be a cyclic representation of $\mathcal A_1 \otimes_{max} \mathcal A_2$. Define
\begin{align*}
    \mathcal M_{1}=\pi(\mathcal A_1\otimes 1), \quad \mathcal M_2=\pi(1\otimes \mathcal A_2)
\end{align*}
and $M_i=C^*(S_i)$, where every element $a\in S_i$ is self-adjoint. Assume that for any generator $a \in S_1$ there exists $b \in \mathcal M_2$ such that $a\ket{v}=b\ket{v}$ and vice versa. Then $\ket{v}$ is a cyclic and separating vector for both $\mathcal M_1$ and $\mathcal M_2$.
\end{lemma}

\begin{proof}
For any word $a_1\dots a_n \in \mathcal M_1$ in the generators there exists  $b_i \in \mathcal M_2, 1\leq i\leq n$ such that $a_1\dots a_n\ket{v}=b_n\dots b_1\ket{v}$. Therefore, we see $\overline{\mathcal M_1\ket{v}}\subseteq\overline{\mathcal M_2\ket{v}}$ and the reverse inclusion follows in the same way. As the representation is cyclic, this implies  $\overline{M_1\ket{v}}=\overline{\mathcal M_2\ket{v}}=H$ so $\ket{v}$ is cyclic for both algebras. 

Assume $a\ket{v}=0$ for some $a\in \mathcal M_1$. Since we have
\begin{align*}
    ab\ket{v}=ba\ket{v}=0
\end{align*}
for all $b \in \mathcal{M}_2$ and we know $\overline{\mathcal M_2\ket{v}}=H$, we deduce $a=0$. Therefore, $\ket{v}$ is separating for $\mathcal M_1$ and similarly for $\mathcal M_2$.
\end{proof}

\begin{proof}[Proof of \cref{thm:comopselftest}]
    We start by establishing the optimal commuting operator correlation for $G_{emb}$ is unique and a commuting operator self-test. Let $\rho$ be an optimal state for the game $G_{\mathrm{emb}}$ with GNS representation $(\pi_\rho,H_{\rho},\ket{v_{\rho}})$. Write
\[ A_a^{x,s}=\pi_\rho(m_a^{(x,s)}\otimes1), \qquad B_b^{y,s}=\pi_\rho(1\otimes n_b^{(y,s)}), \]
where $s\in\{F,t_{\mathrm{CHSH}}\}$, and let $\mathcal M_A,\mathcal M_B$ be the full local measurement algebras.
By \Cref{prop:value}, each of the three tests attains its individual optimal value. In particular, the $\perp$ outcomes on tilted CHSH questions have zero probability: replacing them by a bit would otherwise strictly increase that subgame’s score. Therefore, we have
\[ A_\perp^{x,t}\ket{v_\rho}=B_\perp^{y,t}\ket{v_\rho}=0. \]
For this subgame, consider binary measurements
\[ \widehat A_0^{x,t}=A_0^{x,t}, \qquad \widehat A_1^{x,t}=A_1^{x,t}+A_\perp^{x,t}, \]
and similarly for Bob. For Feige’s game, put $\widehat A_a^{x,F}=A_a^{x,F}$ and likewise for Bob.
For $s\in\{F,t_{\mathrm{CHSH}}\}$, define
\[ \mathcal M_{A,s}=C^*(\{\widehat A_a^{x,s}\}), \qquad \mathcal M_{B,s}=C^*(\{\widehat B_b^{y,s}\}), \qquad K_s=\overline{C^*(\mathcal M_{A,s},\mathcal M_{B,s})\ket{v_{\rho}}}. \]
The restrictions to the Hilbert space $K_s$ give an optimal cyclic model for each subgame. By \Cref{thm:Feige} and \Cref{thm:tCHSH}, these models are unitarily equivalent to the fixed ideal models with vectors
\[ \ket{\Phi_3}=\frac{\ket{00}+\ket{11}+\ket{22}}{\sqrt3}, \qquad \ket{\Psi_2}=\frac{\ket{00}+\sqrt2\,\ket{11}}{\sqrt3}, \]
respectively.
Both ideal vectors have full Schmidt rank, so both models are locally cyclic. Since their local algebras are finite-dimensional and generated by their measurements, for every Alice subgame measurement there is a noncommutative $*$-polynomial $\Gamma_a^{x,s}$ in Bob’s measurement operators such that
\[ \widehat A_a^{x,s}\ket{v_\rho} =\Gamma_a^{x,s}(\{\widehat B_b^{y,s}\})\ket{v_\rho}, \]
and vice versa. These vector identities hold in $K_s\subseteq H_{\rho}$. The same property holds for the original measurements because the additional $\perp$ measurements annihilate $\ket{v_\rho}$. Therefore \Cref{lem:separating} implies
\[ \overline{\mathcal M_A\ket{v_\rho}} =\overline{\mathcal M_B\ket{v_\rho}}=H_{\rho}, \]
and $\ket{v_\rho}$ is separating for both $\mathcal M_A$ and $\mathcal M_B$. Consequently,
\[ A_\perp^{x,t_{\mathrm{CHSH}}}=B_\perp^{y,t_{\mathrm{CHSH}}}=0 \]
as operators on $H_{\rho}$.
We can now lift the subgame operators. For $C\in\{A,B\}$, consider the map
\[ r_{C,s}:\mathcal M_{C,s}\to \mbB(K_s), \qquad r_{C,s}(X)=X|_{K_s}. \]
Since $K_s$ is invariant under every operator in $\mathcal M_{C,s}$ and its adjoint, the restriction map $r_{C,s}$ is a unital $*$-homomorphism. It is injective, since
$r_{C,s}(X)=0$ implies $ X\ket{v_{\rho}}=0$ and we get $X=0$ since $\ket{v_{\rho}}$ is separating.
Thus $r_{C,F}$ and $r_{C,t_{\mathrm{CHSH}}}$ are $*$-isomorphisms onto the local copies of $M_3(\C)$ and $M_2(\C)$, respectively.
We obtain $e^C_{ij}\in \mbB(K_F)$ and $f^C_{uv}\in \mbB(K_{t_{\mathrm{CHSH}}})$ through the unitary equivalences above from the standard local matrix units of the fixed ideal models to \(K_F\) and \(K_{t_{\mathrm{CHSH}}}\). Define their full-space lifts by
\[ \widetilde e^C_{ij}=r_{C,F}^{-1}(e^C_{ij}), \qquad \widetilde f^C_{uv}=r_{C,t}^{-1}(f^C_{uv}). \]
Concretely, if a subgame operator is expressed as
\[ T=\Gamma_a^{x,s}(\{\widehat C_b^{y,s}|_{K_s}\}), \]
then its lift is
\[ \widetilde T=\Gamma_a^{x,s}(\{\widehat C_b^{y,s}\})\in\mathcal M_{C,s}. \]
Injectivity makes this independent of the chosen polynomial expression.
    
   Since the inverse restriction maps are unital $*$-homomorphisms, the lifted operators satisfy the matrix-unit relations on all of $H_{\rho}$.
   Alice’s lifted operators commute with Bob’s, and
\[
\mathcal M_C=C^*(\{\widetilde e^C_{ij}\},\{\widetilde f^C_{uv}\}).
\]
   
   Since $\ket{v_{\rho}}\in K_F \cap K_{t_{CHSH}}$, the ideal-model vector identities give
    \begin{align*}
      \tilde e_{ij}^A\ket{v_\rho}=\tilde e_{ji}^B\ket{v_\rho} \text{ and } \bra{v_{\rho}}\tilde e_{ij}^A\ket{v_\rho}=\bra{v_{\rho}}\tilde e_{ij}^B\ket{v_\rho}=\delta_{ij}\frac 1 3,
    \end{align*}
    and $A_0^{0,F}\ket{v_\rho}=\tilde e_{00}^A\ket{v_\rho}$ for Feige's game. For the tilted CHSH game, we have
  
    \begin{align*}
      \tilde f_{00}^A\ket{v_\rho}=\tilde f_{00}^B\ket{v_\rho}, \tilde f_{10}^B\ket{v_\rho}=\frac{1}{\sqrt{2}} \tilde f_{01}^A\ket{v_\rho}, \tilde f_{10}^A\ket{v_\rho}=\frac{1}{\sqrt{2}} \tilde f_{01}^B\ket{v_\rho} \text{ and } B_0^{0,t_{CHSH}}\ket{v_\rho}=\tilde f_{00}^B\ket{v_\rho}.
    \end{align*}

    \Cref{lem:overlap} now yields $\tilde f_{00}^B\ket{v_\rho}=B_0^{0,t_{CHSH}}\ket{v_\rho}=A_0^{0,F}\ket{v_\rho}=\tilde e_{00}^A\ket{v_\rho}$, from which we deduce $\tilde e_{00}^A\ket{v_\rho}=\tilde f_{00}^{A}\ket{v_\rho}$ and $\tilde e_{00}^B\ket{v_\rho}=\tilde f_{00}^{B}\ket{v_\rho}$. As $\ket{v_{\rho}}$ is  separating by \Cref{lem:separating}, we deduce $\tilde e_{00}^A=\tilde f_{00}^{A}$ and $\tilde e_{00}^B=\tilde f_{00}^{B}$. Therefore
    \begin{align*}
          \mathcal M_A=\pi_{\rho}(A_{POVM}^{X_{emb},A_{emb}}\otimes 1), \mathcal M_B=\pi_{\rho}(1\otimes A_{POVM}^{Y_{emb},B_{emb}})
    \end{align*}
   are generated by elements fulfilling the relations of \Cref{eq:gamealg}, respectively and \Cref{eq:flip} holds. We calculate
   \begin{align*}
       \bra{v_{\rho}}\tilde f_{01}^A\tilde e_{r0}^A \tilde f_{01}^B\tilde e_{s0}^B\ket{v_{\rho}}&=\frac{1}{\sqrt{2}} \bra{v_{\rho}} \tilde e_{r0}^A\tilde f_{10}^B \tilde f_{01}^B\tilde e_{s0}^B\ket{v_{\rho}}\\
       &=\frac{1}{\sqrt{2}} \bra{v_{\rho}} \tilde e_{r0}^A(1-\tilde e_{00}^B)\tilde e_{s0}^B\ket{v_{\rho}}\\
         &=\frac{1}{\sqrt{2}} \bra{v_{\rho}} \tilde e_{0r}^B(1-\tilde e_{00}^B)\tilde e_{s0}^B\ket{v_{\rho}}\\
         &=\frac{1}{\sqrt{2}} \delta_{rs}\bra{v_{\rho}} \tilde e_{00}^B\ket{v_{\rho}}\\
         &=\frac{1}{3\sqrt{2}} \delta_{rs}.
   \end{align*}
   Therefore, \Cref{eq:valuegen} holds. Furthermore, the $*$-polynomials in \Cref{eq:localgebra} exist and are independent of $\rho$, since the ideal models and their standard matrix units are fixed independently of $\rho$. We conclude by \Cref{lem:uniquestate} that the state $\rho$ is unique. Since we proved uniqueness from $\rho$ achieving the optimal value in $G_{emb}$, it follows that the optimal commuting operator correlation $\hat{p}$ is unique.

   Finally, we show that $\hat{p}$ is in $C_{qa}(4,4,3,3)\setminus C_{qs}(4,4,3,3)$. The proof follows essentially from \cite[Theorem 7]{beigi2021separation} with the optimal correlation from Feige's game replacing condition $(i)$ and $p(0,0|(0,F), (0,t_{CHSH}))=p_A(0|(0,F))=p_B(0|(0,t_{CHSH}))=\frac{1}{3}$ replacing condition $(iii)$. Notably, Beigi's proof still works, and we obtain that there is no tensor-product model for this correlation. Lastly, in the proof of \Cref{prop:value}, the strategies in the sequence are finite-dimensional, and their correlations converge to the optimal one.
\end{proof}

\begin{corollary}\label{cor:embez_protocol}
Every optimal cyclic model for $G_{\mathrm{emb}}$ determines an exact embezzlement protocol for
\[|\Phi_2\rangle=\frac{|00\rangle+|11\rangle}{\sqrt2}.\]
The extracted catalyst and local operators are unique up to unitary equivalence.
\end{corollary}

\begin{proof} Let $(\pi_\rho,H_\rho,\ket{v_\rho})$ be an optimal GNS model, with the matrix units from \Cref{thm:comopselftest}. Set
\[
P^C=\widetilde e_{00}^C=\widetilde f_{00}^C,\qquad E=P^AP^B,\qquad K=EH_\rho,\qquad \ket{\xi}=\sqrt3E\ket{v_\rho}.
\]
Since $P^A\ket{v_\rho}=P^B\ket{v_\rho}$ and $\bra{v_\rho}P^A\ket{v_\rho}=\frac{1}{3}$, we have $\|\xi\|=1$.
For $C\in\{A,B\}$ and $r\in\{1,2\}$, define
\[
V_r^C=\widetilde f_{01}^C\widetilde e_{r0}^C.\]
The subspace $K$ is invariant under every $V_r^C$ and its adjoint. By \Cref{lem:isoCuntz}, it holds
\[
(V_r^C)^*V_s^C=\delta_{rs}P^C \text{ and }\sum_{r=1}^2V_r^C(V_r^C)^*=P^C.
\]
Consequently, the operators
\[r_i=(V_{i+1}^A)^*|_K,\qquad t_i=(V_{i+1}^B)^*|_K,\qquad i\in\{0,1\},
\]
generate commuting unital $C^*$-algebras $\mathcal A=C^*(r_0,r_1)$ and $\mathcal B=C^*(t_0,t_1)$, and satisfy
\[
\sum_{i=0}^1 r_i^*r_i=\sum_{i=0}^1 t_i^*t_i=I_K.
\]
Let $\omega$ be the state on $\mathcal A \otimes_{max} \mathcal B$ defined by $\omega(a\otimes b)=\bra{\xi}ab\ket{\xi}$. Since it holds $V_r^AV_s^B=EV_r^AV_s^BE$, \Cref{eq:valuegen} gives
\[
\omega(r_i\otimes t_j)=3\bra{v_{\rho}}(V_{i+1}^A)^*(V_{j+1}^B)^*\ket{v_{\rho}}=\frac{\delta_{ij}}{\sqrt{2}}.
\]
Thus $\omega$, together with $(r_0,r_1)$ and $(t_0,t_1)$ is an exact embezzlement protocol for $\ket{\Phi_2}$ in the sense of \Cref{def:abs_embezzlement}.

Indeed, it holds that
\[
\sum_{i,j=0}^1\|r_it_j\ket{\xi}-\frac{\delta_{ij}}{\sqrt{2}}\ket{\xi}\|^2=\sum_{i,j=0}^1\|r_it_j\ket{\xi}\|^2-1=0,
\]
and hence $r_it_j\ket{\xi}=\frac{\delta_{ij}}{\sqrt{2}}\ket{\xi}$. By \Cref{prop:equiv_def}, adjoining local ancillas yields commuting unitaries implementing
\[
\ket{0}\otimes \ket{\hat\xi}\otimes \ket{0} \mapsto \frac{1}{\sqrt{2}}(\ket{0}\otimes \ket{\hat\xi}\otimes \ket{0}+\ket{1}\otimes \ket{\hat\xi}\otimes \ket{1}),
\]
where $\ket{\hat\xi}$ is $\ket{\xi}$ tensored with the initial ancillas. 

Finally, \Cref{thm:comopselftest} gives a unitary equivalence between any two optimal cyclic models. Since the matrix units are fixed polynomials in the measurement operators, this equivalence also intertwines $E$, $\ket{\xi}$, and the operators $r_i,t_j$. This proves the uniqueness asserted in the statement.
\end{proof}

\bibliography{references}
\bibliographystyle{alpha}

\end{document}

%% file: ams_header.tex
\usepackage{amsmath,amsthm,amsfonts,amssymb,latexsym,verbatim,bbm,hyperref,mathrsfs}
\usepackage{mathtools}
\usepackage[shortlabels]{enumitem}
\usepackage{subcaption}
\usepackage{thmtools}
\usepackage[capitalize]{cleveref}
\usepackage{afterpage}
\usepackage{xcolor}
\usepackage{braket}

\DeclareRobustCommand{\SkipTocEntry}[5]{}

\allowdisplaybreaks

\setlist{itemsep=.5\baselineskip,topsep=.5\baselineskip}

\numberwithin{equation}{section}
\theoremstyle{plain}
\newtheorem{theorem}{Theorem}[section]

\newtheorem{lemma}[theorem]{Lemma}
\newtheorem{proposition}[theorem]{Proposition}

\newtheorem{definition}[theorem]{Definition}

\newtheorem{remark}[theorem]{Remark}
\newtheorem{corollary}[theorem]{Corollary}

\newtheorem*{theoremI}{Theorem}
\newtheorem*{corI}{Corollary}

\newcommand{\ang}[1]{\langle #1 \rangle}
\newcommand{\R}{\mathbb{R}}
\newcommand{\C}{\mathbb{C}}

\newcommand{\Id}{\mathbbm{1}}

\newcommand{\mcA}{\mathcal{A}}
\newcommand{\mcB}{\mathcal{B}}

\newcommand{\mcF}{\mathcal{F}}
\newcommand{\mcG}{\mathcal{G}}

\newcommand{\mcI}{\mathcal{I}}

\newcommand{\mcM}{\mathcal{M}}

\newcommand{\mcO}{\mathcal{O}}

\newcommand{\mcR}{\mathcal{R}}
\newcommand{\mcS}{\mathcal{S}}
\newcommand{\mcT}{\mathcal{T}}
\newcommand{\mcU}{\mathcal{U}}

\newcommand{\mbB}{\mathbb{B}}

\newcommand{\onto}{\twoheadrightarrow}

\newcommand{\tr}{\text{tr}}
\newcommand{\wtd}{\widetilde}

%% file: references.bib
@article{coladangelo2020two,
  title={A two-player dimension witness based on embezzlement, and an elementary proof of the non-closure of the set of quantum correlations},
  author={Coladangelo, Andrea},
  journal={Quantum},
  volume={4},
  pages={282},
  year={2020},
  publisher={Verein zur F{\"o}rderung des Open Access Publizierens in den Quantenwissenschaften}
}

@article{schmidt2025guess,
  title={Guess your neighbor's input: Quantum advantage in {F}eige's game},
  author={Schmidt, Simon and Storgaard, Sigurd AL and Walter, Michael and Zhao, Yuming},
  journal={arXiv preprint arXiv:2510.08484},
  year={2025}
}

@article{PSZZ,
  title={An operator-algebraic formulation of self-testing},
  author={Paddock, Connor and Slofstra, William and Zhao, Yuming and Zhou, Yangchen},
  journal={Annales Henri Poincar{\'e}},
  volume={25},
  number={10},
  pages={4283--4319},
  year={2024},
  organization={Springer}
}

@article{CTT,
  title={An operator system approach to self-testing},
  author={Crann, Jason and Todorov, Ivan G and Turowska, Lyudmila},
  journal={Advances in Mathematics},
  volume={492},
  pages={110884},
  year={2026},
  publisher={Elsevier}
}

@article{harris2026self,
  title={Self-testing of exact entanglement embezzlement},
  author={Harris, Samuel J},
  journal={arXiv preprint arXiv:2605.22713},
  year={2026}
}

@article{van2025multipartite,
  title={Multipartite embezzlement of entanglement},
  author={van Luijk, Lauritz and Stottmeister, Alexander and Wilming, Henrik},
  journal={Quantum},
  volume={9},
  pages={1818},
  year={2025},
  publisher={Verein zur F{\"o}rderung des Open Access Publizierens in den Quantenwissenschaften}
}

@inproceedings{BCKLMNS,
  title={A mathematical foundation for self-testing: Lifting common assumptions},
  author={Baptista, Pedro and Chen, Ranyiliu and Kaniewski, Jedrzej and Lolck, David Rasmussen and Man{\v{c}}inska, Laura and Nielsen, Thor Gabelgaard and Schmidt, Simon},
  booktitle={Annales Henri Poincar{\'e}},
  pages={1--48},
  year={2025},
  organization={Springer}
}

@article{liu2025embezzlement,
  title={Embezzlement as a `Self-Test' for Infinite Copies of Entangled States},
  author={Liu, Li},
  journal={arXiv preprint arXiv:2509.05036},
  year={2025}
}

@article{liu2025explicit,
  title={Explicit {C}*-algebraic Protocol for Exact Universal Embezzlement of Entanglement},
  author={Liu, Li},
  journal={arXiv preprint arXiv:2506.10736},
  year={2025}
}

@article{cleve2017perfect,
  title={Perfect embezzlement of entanglement},
  author={Cleve, Richard and Liu, Li and Paulsen, Vern I},
  journal={Journal of Mathematical Physics},
  volume={58},
  number={1},
  year={2017},
  publisher={AIP Publishing}
}

@article{van2003universal,
  title={Universal entanglement transformations without communication},
  author={van Dam, Wim and Hayden, Patrick},
  journal={Physical Review A},
  volume={67},
  number={6},
  pages={060302},
  year={2003},
  publisher={APS}
}

@inproceedings{slofstra2019set,
  title={The set of quantum correlations is not closed},
  author={Slofstra, William},
  booktitle={Forum of Mathematics, Pi},
  volume={7},
  pages={e1},
  year={2019},
  organization={Cambridge University Press}
}

@article{bamps2015sum,
  title={Sum-of-squares decompositions for a family of {C}lauser-{H}orne-{S}himony-{H}olt-like inequalities and their application to self-testing},
  author={Bamps, C{\'e}dric and Pironio, Stefano},
  journal={Physical Review A},
  volume={91},
  number={5},
  pages={052111},
  year={2015},
  publisher={APS}
}

@inproceedings{mayers1998quantum,
  title={Quantum cryptography with imperfect apparatus},
  author={Mayers, Dominic and Yao, Andrew},
  booktitle={Proceedings 39th Annual Symposium on Foundations of Computer Science (Cat. No. 98CB36280)},
  pages={503--509},
  year={1998},
  organization={IEEE}
}

@article{mayers2003self,
  title={Self testing quantum apparatus},
  author={Mayers, Dominic and Yao, Andrew},
  journal={Quantum Information \& Computation},
  volume={4},
  number={4},
  pages={273--286},
  year={2004},
  publisher={Rinton Press, Incorporated Paramus, NJ},
doi={10.26421/QIC4.4-3},
}

@INPROCEEDINGS{Feige1991,
  author={Feige, U.},
  booktitle={Proceedings of the Sixth Annual Structure in Complexity Theory Conference}, 
  title={On the success probability of the two provers in one-round proof systems}, 
  year={1991},
  volume={},
  number={},
  pages={116-123},
  doi={10.1109/SCT.1991.160251}}

@book {murphy,
    AUTHOR = {Murphy, Gerard J.},
     TITLE = {{$C^*$}-algebras and operator theory},
 PUBLISHER = {Academic Press, Inc., Boston, MA},
      YEAR = {1990},
     PAGES = {x+286},
      ISBN = {0-12-511360-9},
   MRCLASS = {46Lxx (46-01)},
  MRNUMBER = {1074574},
MRREVIEWER = {E.\ Gerlach},
}

@article{beigi2021separation,
  title={Separation of quantum, spatial quantum, and approximate quantum correlations},
  author={Beigi, Salman},
  journal={Quantum},
  volume={5},
  pages={389},
  year={2021},
  publisher={Verein zur F{\"o}rderung des Open Access Publizierens in den Quantenwissenschaften}
}

@article{ji2021mip,
  title={{MIP}*={RE}},
  author={Ji, Zhengfeng and Natarajan, Anand and Vidick, Thomas and Wright, John and Yuen, Henry},
  journal={Communications of the ACM},
  volume={64},
  number={11},
  pages={131--138},
  year={2021},
  publisher={ACM New York, NY, USA}
}

@article{KSZ26,
  title={Sums of projections: quantitative stability and infinite-dimensional self-testing},
  author={Kar, Prem Nigam and Schafhauser, Christopher and Zhao, Yuming},
  journal={in preparation},
  year={2026}
}

@article{paulsen2016estimating,
  title={Estimating quantum chromatic numbers},
  author={Paulsen, Vern I and Severini, Simone and Stahlke, Daniel and Todorov, Ivan G and Winter, Andreas},
  journal={Journal of Functional Analysis},
  volume={270},
  number={6},
  pages={2188--2222},
  year={2016},
  publisher={Elsevier}
}

@article{reichardt2013classical,
  title={Classical command of quantum systems},
  author={Reichardt, Ben W and Unger, Falk and Vazirani, Umesh},
  journal={Nature},
  volume={496},
  number={7446},
  pages={456--460},
  year={2013},
  publisher={Nature Publishing Group UK London}
}

@article{fritz2012tsirelson,
  title={Tsirelson's problem and {K}irchberg's conjecture},
  author={Fritz, Tobias},
  journal={Reviews in Mathematical Physics},
  volume={24},
  number={05},
  pages={1250012},
  year={2012},
  publisher={World Scientific}
}
